\documentclass{article}

\PassOptionsToPackage{table}{xcolor}
\usepackage{iclr2025_conference,times}

\usepackage{amsmath,amsfonts,bm}

\def\eqref#1{\textup{(\ref{#1})}}

\def\1{\bm{1}}

\DeclareMathAlphabet{\mathsfit}{\encodingdefault}{\sfdefault}{m}{sl}
\SetMathAlphabet{\mathsfit}{bold}{\encodingdefault}{\sfdefault}{bx}{n}

\newcommand{\E}{\mathbb{E}}

\newcommand{\R}{\mathbb{R}}

\usepackage{amsmath}
\usepackage{amssymb}
\usepackage{amsthm}
\usepackage{booktabs}
\usepackage{enumitem}
\usepackage{graphicx}
\usepackage{multirow}
\usepackage{xcolor}
\usepackage[hidelinks]{hyperref}
\usepackage{url}

\definecolor{oursbg}{gray}{0.92}

\newcommand{\Tend}{T}                      
\newcommand{\per}{t}                       

\newcommand{\dem}{d}                       
\newcommand{\ord}{q}                       
\newcommand{\lead}{L}                      
\newcommand{\arr}{A}                       
\newcommand{\sold}{s}                      
\newcommand{\onh}{I}                       
\newcommand{\onhend}{\onh^{\mathrm{end}}}  
\newcommand{\intr}{X}                      
\newcommand{\IP}{\mathrm{IP}}              
\newcommand{\prof}{p}                      
\newcommand{\hold}{h}                      
\newcommand{\crat}{\rho}                     
\newcommand{\rew}{R}                       

\newcommand{\Hist}{\mathcal{H}}            

\newcommand{\Short}{U}                     
\newcommand{\Iproj}{\widehat{\onh}}        
\newcommand{\Lost}{\Lambda}                
\newcommand{\dep}{\kappa}                  
\newcommand{\rate}{r}                      
\newcommand{\lev}{\theta}                  
\newcommand{\capm}{c}                       
\newcommand{\cbuy}{c_{\mathrm{b}}}          
\newcommand{\psurv}{\pi}                   
\newcommand{\base}{y}                      

\newcommand{\Mset}{\mathcal{M}}            
\newcommand{\Msize}{M}                     
\newcommand{\wt}{w}                        
\newcommand{\lw}{\ell}                     
\newcommand{\lss}{\mathsf{l}}              
\newcommand{\fgt}{\gamma}                  
\newcommand{\Spaces}{\mathcal{E}}          
\newcommand{\tf}{\phi}                     
\newcommand{\dof}{\nu}                     
\newcommand{\npath}{K}                     
\newcommand{\hor}{H}                       

\newcommand{\shp}{\eta}                    
\newcommand{\sexp}{\alpha}                 
\newcommand{\nvar}{G}                      

\newcommand{\Norm}{\mathrm{NR}}            
\newcommand{\Score}{\mathrm{S}}            

\newcommand{\Fpos}{\mathcal{B}}            
\newcommand{\Frate}{\mathcal{R}}           
\newcommand{\Fcap}{\mathcal{K}}            

\newcommand{\pos}[1]{\left(#1\right)^{+}}
\newcommand{\ind}[1]{\mathbf{1}\!\left\{#1\right\}}
\newcommand{\Prob}{\mathbb{P}}

\newcommand{\ourscore}{0.6311}
\newcommand{\pubbest}{0.5380}
\newcommand{\ceiling}{0.8656}
\newcommand{\hindsight}{0.5892}
\newcommand{\npublished}{ten}
\newcommand{\execgap}{-0.0006}
\newcommand{\deadsyn}{32\%}
\newcommand{\deadreal}{30\%}
\newcommand{\ceilsynstoch}{0.6801}
\newcommand{\ceilrealstoch}{0.7039}
\newcommand{\hsleadsyn}{0.6220}
\newcommand{\hsleadreal}{0.4185}
\newcommand{\bestleadsyn}{0.5576}
\newcommand{\bestleadreal}{0.3085}

\newcommand{\dtcalls}{674}
\newcommand{\ddcalls}{64,200}
\newcommand{\dtratio}{95}
\newcommand{\shapeitems}{200}
\newcommand{\armcalls}{12,405}

\newcommand{\abinter}{-0.0016}
\newcommand{\abrpooled}{+0.0010}
\newcommand{\abseedsd}{0.0001}
\newcommand{\abcore}{0.6180}
\newcommand{\abauthored}{+0.0138}
\newcommand{\abshare}{15\%}
\newcommand{\abshapewhere}{+0.0179}
\newcommand{\absearchwhere}{+0.0230}
\newcommand{\abrsearch}{+0.0007}
\newcommand{\abrshape}{+0.0022}

\newcommand{\ctrln}{124}



\theoremstyle{plain}
\newtheorem{theorem}{Theorem}
\newtheorem{lemma}[theorem]{Lemma}
\newtheorem{proposition}[theorem]{Proposition}
\newtheorem{corollary}[theorem]{Corollary}
\theoremstyle{definition}
\newtheorem{assumption}{Assumption}

\newtheorem{example}{Example}
\theoremstyle{remark}
\newtheorem{remark}{Remark}

\newcommand{\sabremark}[1][0.86]{%
  \raisebox{-0.14ex}{\includegraphics[height=#1em]{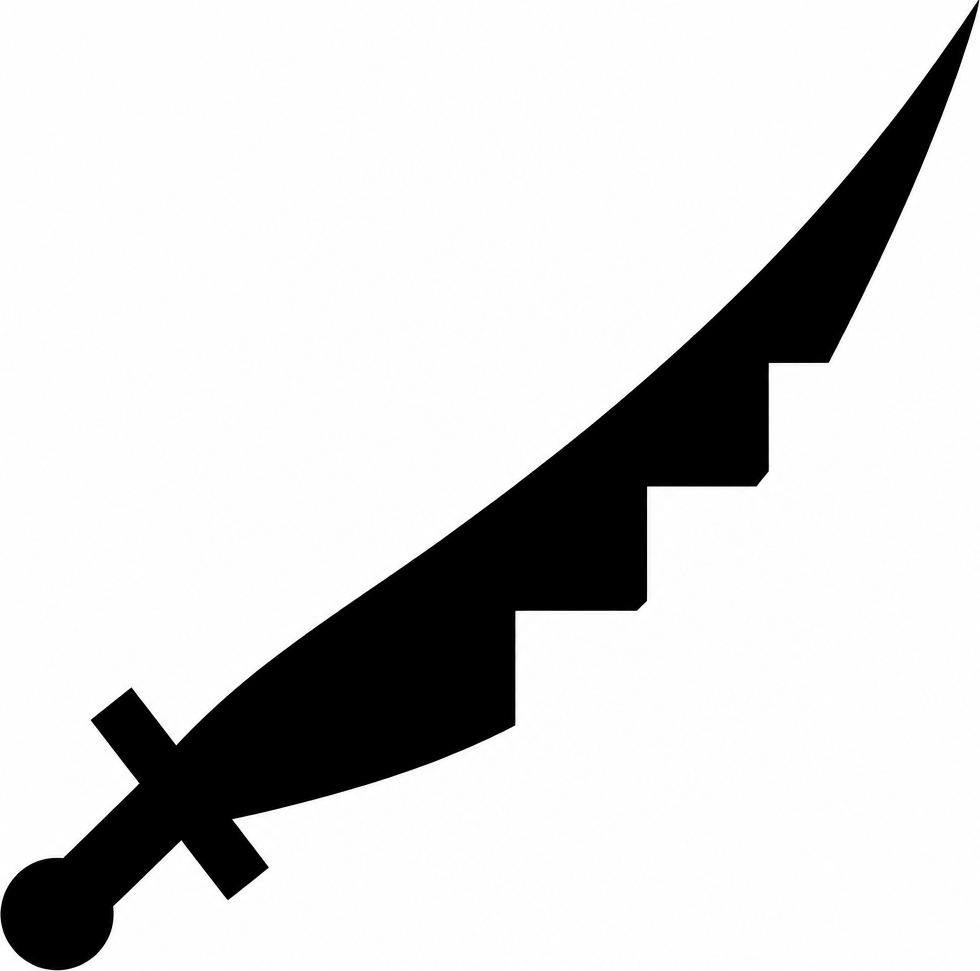}}\hspace{0.24em}}

\title{\sabremark[1.12]SabreAgent: Language Models at Design Time for Lost-Sales Inventory Control}

\author{\textbf{Yang Liu}\textsuperscript{1},
\textbf{Yulin Huang}\textsuperscript{1},
\textbf{Xue Yu}\textsuperscript{1},
\textbf{Jiong Dong}\textsuperscript{1},
\textbf{Jianshen Zhang}\textsuperscript{1},
\textbf{Yongzhi Qi}\textsuperscript{1}$^{\dagger}$
\\[0.6em]
\textsuperscript{1}Supply Chain Tech Team Y, JD.com\\
$^{\dagger}$Corresponding author.
}

\iclrfinalcopy

\begin{document}

\maketitle

\begin{abstract}
SabreAgent uses a language model at design time to construct two components for lost-sales inventory control: a product-specific seasonal prior and a validation-selected capped base-stock policy family. During operation, statistical forecasting and inventory optimization use these frozen artifacts to determine orders, with zero language-model calls. We evaluate the approach on the $1{,}320$ instances of InventoryBench. Under the benchmark's cost assumptions, the operations-research core draws on a zero-lead-time optimality result and a projected-inventory rule for positive deterministic lead times. The latter computes replenishment shortfalls by propagating inventory using sales along simulated demand paths. The seasonal prior adds forecast variants alongside the original forecaster, and the selected policy family handles stochastic lead times with order destruction. SabreAgent scores $0.6311$, compared with $0.5380$ for the strongest published baseline, and ranks first in all six benchmark cells. Ablations attribute most of the gain to the OR core. In the paired analysis, the seasonal component adds $1.79\%$ across the three real-data cells, and the search component adds $2.3\%$ across the two stochastic-lead-time cells. These results demonstrate how model-generated priors and policy structure can improve an OR controller through design-time use.
\end{abstract}

\section{Introduction}
\label{sec:intro}

InventoryBench gives each policy only five initial demand observations, while real instances also provide product descriptions \citep{inventorybench}. Effective replenishment requires combining these limited observations and seasonal cues with inventory costs and supply dynamics. Nine of the benchmark's ten published baselines consult a language model every period. Each resulting order commits inventory before its outcome is observed. Work on planning and self-correction identifies limitations of using generated reasoning to assess such decisions \citep{huang2024selfcorrect,liu2025newsvendor}.

Design-time use lets external procedures inspect and evaluate a model-generated
policy, prior, or representation before deployment. This division of labor
underlies verifier-guided program search
\citep{romera2024mathematical,liu2024evolution,ye2024reevo,novikov2025alphaevolve}
and the broader ``model proposes, verifier disposes'' view
\citep{kambhampati2024position}. Prior work in operations research has used
language models to formulate and translate optimization problems
\citep{ramamonjison2023nl4opt,ahmaditeshnizi2024optimus,xiao2024chainofexperts}.
SabreAgent applies this division of labor to inventory control, with statistical
forecasting and optimization handling sequential decisions during operation.
Figure~\ref{fig:motivation} illustrates this division of labor.

\begin{figure}[t]
\centering
\includegraphics[width=\textwidth]{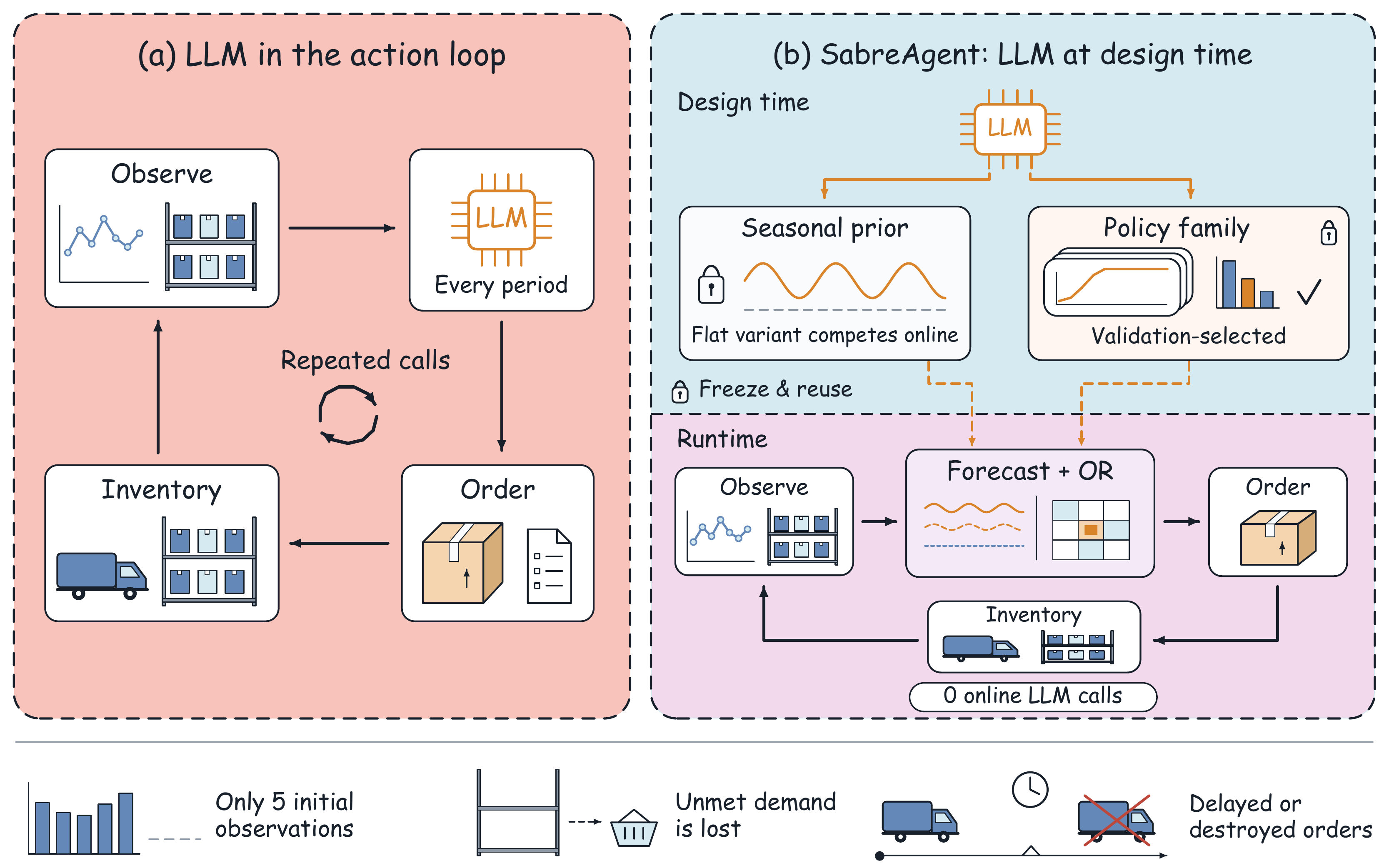}
\caption{Two placements of language models in inventory control. (a) A model
generates orders inside the decision loop. (b) SabreAgent uses a frozen seasonal
prior alongside the forecaster without seasonal adjustment and a
validation-selected policy family.
Statistical forecasting and OR control generate orders without runtime
language-model calls.}
\label{fig:motivation}
\end{figure}

Across InventoryBench's $1{,}320$ lost-sales instances, all ten published
baselines score below a hindsight diagnostic that chooses one constant
order-up-to level per instance. This offline reference motivates an explicit
treatment of inventory costs and dynamics.

SabreAgent builds its control core from the benchmark's cost structure. Orders
have no purchase or setup cost and leftover inventory has no salvage value. At
zero lead time, some optimal post-order inventory level is bounded above by
$\max\{\onh_\per,\base^*_\per\}$, where $\base^*_\per$ is the myopic
critical-fractile level. Under stationary demand, the myopic policy is optimal.
For positive deterministic lead times, we project the committed
pipeline through the lost-sales dynamics and price only the residual shortfall.
A textbook inventory-position target treats demand that was
already lost as if it had depleted the pipeline. For stochastic lead times, where
an order may be destroyed, the policy instead evaluates a small family of
replenishment rules using the same simulated demand and supply trajectories.

The language model supplies two additions to this core. We elicit a normalized
seasonal shape from product descriptions and use it to construct additional
forecast variants. The original forecaster without seasonal adjustment remains
in the mixture; $\S$\ref{sec:shape} gives the associated predictive-loss bounds.
An automated research loop also proposes policy families for the lossy-supply
cells. Scores on separately generated validation instances guide the selection
of a capped base-stock family. Building on
established capped base-stock policies in lost-sales control
\citep{xin2021understanding}, we link the cap to destroyed-order catch-up bursts
and combine design-time family selection with online, instance-specific parameter
search.

The final policy scores \ourscore{}, versus \pubbest{} for the strongest
published baseline, and leads in every benchmark cell. The ablation separates
the sources of this performance:
the frozen OR core reaches \abcore{} on the batched evaluation path, and the two
model-authored components add \abauthored{} on that same
path. Paired bootstrap intervals exclude zero for the seasonal component on
real-data cells and the search component on stochastic-lead-time cells.
A same-model online-adjustment control adds bounded per-period order changes
to the complete policy. On its \ctrln{}-instance subset, the real stochastic
cell is the only cell with a positive estimate exceeding two standard errors.

The analysis and ablations identify the contributions of the cost-based control
core, seasonal forecasting, and policy-family selection. Related work and
proofs appear in the appendix.

\section{Preliminaries}
\label{sec:prelim}

\subsection{Control problem and evaluation}
\label{sec:prelim-problem}

Each InventoryBench instance \citep{inventorybench} tracks one product at one
location for $\Tend$ weeks, with $\Tend=50$ for synthetic instances and $47$ for
real ones. Inventory starts at zero. The nonnegative demand path $\dem_\per$ is
fixed by the instance, and unmet demand is lost
\citep{zipkin2008old,bijvank2011lost}. In the benchmark's own timing the policy
first observes on-hand stock $\onh_\per$, the
in-transit total $\intr_\per$ and the position $\IP_\per = \onh_\per + \intr_\per$,
all before that period's arrivals are booked. It then places
$\ord_\per \in \mathbb{Z}_{\ge 0}$, scheduled for $\per + \lead_\per$ when
$\lead_\per < \infty$ and destroyed otherwise. Arrivals
$\arr_\per = \sum_{\per' \le \per} \ord_{\per'}
\ind{\lead_{\per'} < \infty,\ \per' + \lead_{\per'} = \per}$ are added to stock
next, the sum including $\per' = \per$, so at $\lead = 0$ an order lands in the
period that places it. Sales are $\sold_\per = \min(\dem_\per, \onh_\per +
\arr_\per)$, and holding is charged on the ending inventory $\onhend_\per =
\onh_\per + \arr_\per - \sold_\per$, which also carries over:
\begin{equation}
  \onh_{\per+1} = \onhend_\per,
  \qquad
  \rew_\per = \prof\,\sold_\per - \hold\,\onhend_\per .
  \label{eq:prelim-reward}
\end{equation}
Costs are $(\prof,\hold) \in \{(1,1),(4,1),(19,1)\}$, so the critical ratio
$\crat = \prof/(\prof+\hold)$ is $0.50$, $0.80$ or $0.95$, all three present
in every cell.

The policy receives a partial observation of the system state. Initially it receives five
dated demand observations, the promised lead time, $\prof$, $\hold$, and, for a
real instance, a product description. In later periods it sees the date,
$\onh_\per$, $\intr_\per$, and the previous demand, order, and arrival, then maps
the resulting history $\Hist_\per$ to $\ord_\per$. In-transit inventory is only a
total; its age profile and the current realization of $\lead_\per$ are hidden when
the order is placed. Demand itself is uncensored, unlike the usual lost-sales
setting in which stockouts conceal the upper tail and couple learning to control
\citep{huh2009nonparametric}. No exploration term is therefore needed for demand
estimation.

Two regimes set $\lead_\per$ to $0$ or $4$. The stochastic regime uses
$\{1,2,3,\infty\}$; $\lead_\per=\infty$ denotes permanent order destruction.
The benchmark fixes one lead-time sequence per cell and shares it
across that cell's instances. A policy is told only a promised lead time of $2$
and must learn survival from observed inventory movements. Because demand and
supply paths are fixed during scoring, the expectation in the policy's planning
objective is internal to its predictive model; the reported score is computed on
the realized benchmark path. Crossing two data sources with three lead-time
regimes gives the six cells in Table~\ref{tab:app-cells}.

Each instance is scored by a normalized reward, and the headline number is its
unweighted mean over the $1{,}320$ instances:
\begin{equation}
  \Norm = \max\!\left\{0,\;
    \frac{\sum_{\per} \rew_\per}{\prof \sum_{\per} \dem_\per} \right\},
  \qquad
  \Score = \frac{1}{1{,}320} \sum_{i=1}^{1{,}320} \Norm_i .
  \label{eq:prelim-score}
\end{equation}
The denominator is the reward from serving every unit immediately with no ending
inventory. It is attainable with perfect demand information when $\lead=0$.
With positive lead time, however, an empty initial system cannot serve demand
before its first possible arrival, and destroyed orders impose an additional
loss. Thus $\Norm=1$ is the ideal-service reference for normalization;
$\S$\ref{sec:exp-ceiling} reports a supply-feasible clairvoyant diagnostic.

\subsection{What the cost structure implies}
\label{sec:prelim-cost}

Purchase cost, fixed ordering cost, and salvage value are all zero in the
benchmark. Inventory can normally be valuable because it avoids a future purchase
or setup cost \citep{scarf1960optimality}, or because it covers demand that no
future order can reach
\citep{arrow1951optimal, zipkin2000foundations, porteus2002foundations}. Under
\eqref{eq:prelim-reward}, a future order is free. At $\lead=0$, it can also
arrive in time to serve that period's demand. Carrying an extra unit therefore
offers no saving on future orders and incurs $\hold$ in each period it remains
unsold. The analysis takes $\prof, \hold > 0$, orders
unbounded above, and demand exogenous, nonnegative and independent across periods
with known marginals.

\begin{theorem}[zero-lead-time policy structure]
\label{thm:zero-option}
Assume zero purchase, setup, and salvage costs and $\lead=0$. Let
$\base^*_\per$ be a $\crat$-fractile of period-$\per$ demand. Then $V_\per$ is
non-increasing, and some optimal post-order level is no greater than
$\max\{\onh_\per,\base^*_\per\}$. Lookahead may reduce the myopic order but cannot
improve the value by increasing it. If the fractiles admit a non-decreasing
selection over time---in particular, under stationary demand---then $V_\per$ is
constant on $[0,\base^*_\per]$ and the myopic policy
$\base_\per=\max\{\onh_\per,\base^*_\per\}$ is optimal.
\end{theorem}

The proof, integer-order extension, and a counterexample when fractiles fall over
time appear in Appendix~\ref{app:proofs}. Classical work studies near-myopic
behavior in broader inventory settings
\citep{karlin1958inventory,morton1971near,zipkin2008structure}. The exact statement
here uses instantaneous, free replenishment. Our positive
lead-time heuristic applies the theorem's cost intuition through a
projected-inventory surrogate.

Zero salvage also fixes the rollout boundary condition: inventory left after the
last scored period has no terminal value. The deterministic rule below therefore
prices the marginal service benefit of an order against the holding cost of its
projected surplus, without an added continuation credit.

\section{\texorpdfstring{\sabremark[1.00]SabreAgent}{SabreAgent}}
\label{sec:method}

\subsection{Design-time architecture}
\label{sec:overview}
\label{sec:loop}

\begin{figure}[t]
\centering
\includegraphics[width=\textwidth]{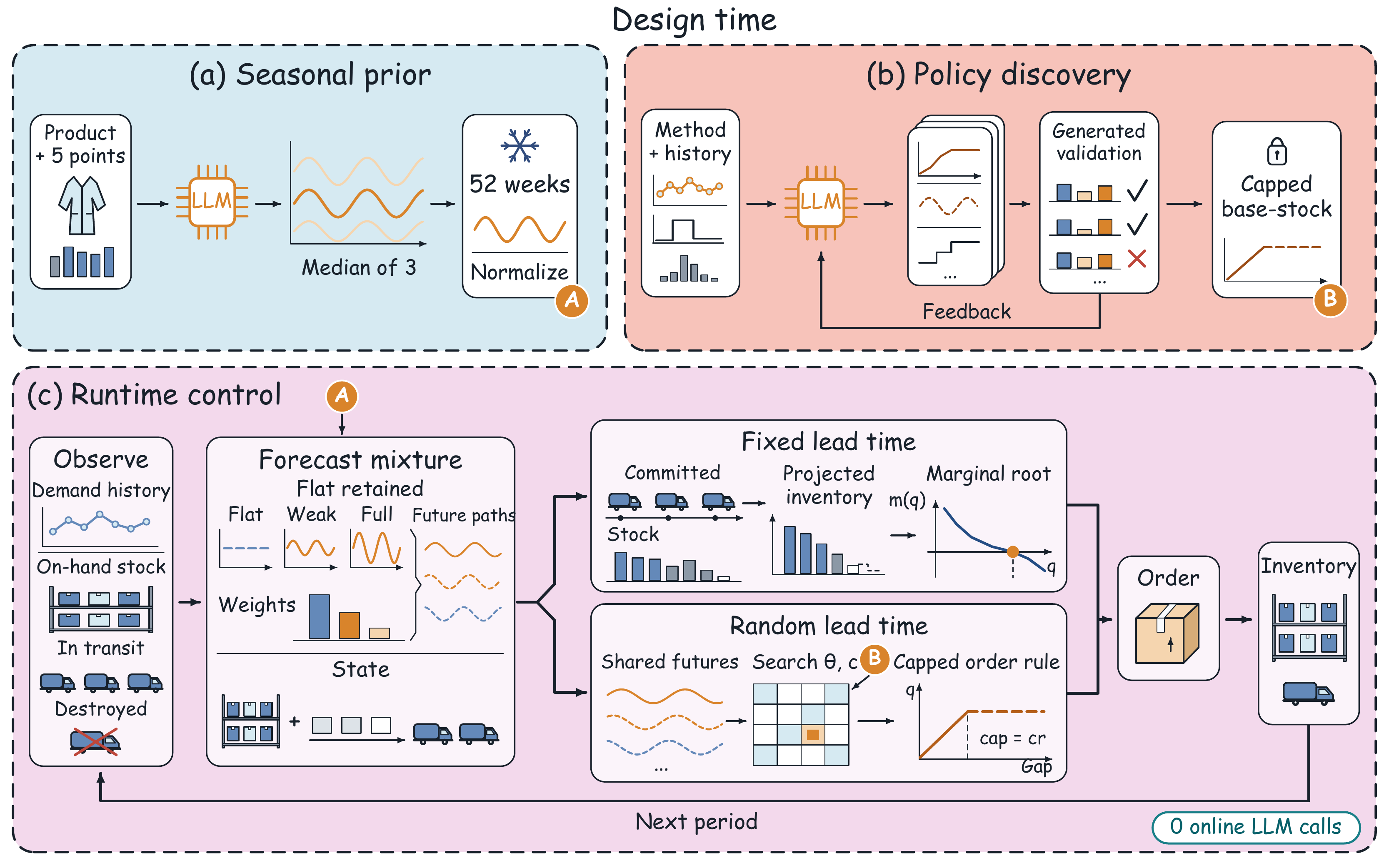}
\caption{SabreAgent's two-stage architecture. (A) The frozen seasonal prior
enters a forecast mixture that retains the flat variant; the associated
bounds concern predictive log loss. (B) Validation
selects the frozen capped base-stock family, whose numerical parameters are
searched online. Fixed and stochastic lead times use projected-inventory control
and rollouts on common demand and supply paths, respectively. Runtime uses no language-model calls.}
\label{fig:framework}
\end{figure}

SabreAgent has three independently configurable components: a forecaster
($\S$\ref{sec:forecaster}), an optional seasonal-shape layer
($\S$\ref{sec:shape}), and the policy family searched when lead times are
stochastic ($\S$\ref{sec:stoch}). We designed the forecaster and the initial
policy; DeepSeek-V4-Flash produced the other two artifacts before the final
configuration was frozen. All language-model calls occur before deployment.
During operation, the policy uses the frozen artifacts to update forecasts and
search policy parameters online; a fixed random seed makes these computations
reproducible. Restoring all
components to their initial settings gives the frozen core, enabling component
ablations within the same policy implementation (Appendix~\ref{app:method}).

The two model-authored components use different selection procedures. For policy search, an
automated loop receives the current method, the validation protocol, and the full
record of earlier successes and failures. It proposes one self-contained policy
change per round. Candidates advance according to scores from the validation
harness. Candidate scoring uses only generated validation instances separate from
the benchmark test instances.
Section~\ref{sec:exp-protocol} records the benchmark information available during
development and the resulting limits on test independence.

Seasonal shapes are frozen per item before benchmark evaluation. The validation
instances lack product text and therefore provide no held-out target for selecting
these shapes. The shape layer adds forecast variants beside the unchanged
flat variant; $\S$\ref{sec:shape} states the corresponding log-loss bounds
and their conditions.

\subsection{Online forecasting in two spaces}
\label{sec:forecaster}

Each instance begins with only five observations, so the forecast bank $\Mset$
uses simple, low-parameter members and learns their mixture weights online. This
choice follows the strong small-sample performance of simple forecasts and
forecast combinations \citep{makridakis2020m4,hyndman2008automatic}. We use
discounted Bayesian model averaging
\citep{hoeting1999bayesian,raftery2010online}, with a forgetting factor $\fgt$
that discounts past fit geometrically; Appendix~\ref{app:method} lists the nine
members of one bank and gives the weight update.

For deterministic lead times, nine members are fit in both raw and
$\tf(u)=\log(1+u)$ space, giving $\Msize=18$. The transform can better represent
right-skewed weekly demand, while the raw-space members retain alternative shapes
when only five observations are available \citep{box1964analysis}. The online weights choose
between them. Stochastic-lead-time cells use only the raw-space bank because their
decision rule takes a local demand rate as its forecast input.

We compare members across transform spaces through their raw-space scoring
densities, using a change-of-variables correction. A log-space member receives raw-space log density
$\log q(\log(1+\dem_\per))-\log(1+\dem_\per)$; the Jacobian term is absent for a
raw-space member. Omitting it persistently biases the log-space weights upward
relative to the raw-space weights. Lemma~\ref{lem:jacobian} quantifies
this bias. For both decision rules, the demand-path sampler uses the learned
scoring weights to select a member; its within-member sampling distribution is
specified separately. Member
definitions, scoring densities, and sampling details are in Appendix~\ref{app:method}.

\subsection{Product text as a seasonal prior}
\label{sec:shape}

Product descriptions in the real instances provide seasonal timing cues, such as
``winter coat'', before those patterns appear in the short demand history.
Synthetic instances lack this text. At design time we ask
DeepSeek-V4-Flash for a $52$-week seasonal curve $\shp$ with mean one, used as
a multiplicative demand factor. We draw three curves per product and aggregate
the accepted, normalized samples by their week-wise median. The aggregate is
then clipped, renormalized to mean one, and frozen before evaluation
(Appendix~\ref{app:method}). The curve encodes relative timing only; the online forecaster
continues to estimate the demand level.

Our validation generator supplies neither product text nor a real calendar for
selecting prior strength. We therefore carry three strengths,
$\sexp\in\{0,1/2,1\}$. Each runs a copy of the bank on
demand deseasonalized by $\shp^{\sexp_g}$, and the variants are weighted online
by the same discounted predictive likelihood used within the bank. The associated
Jacobian keeps their raw-space scoring densities comparable (Appendix~\ref{app:method}).

The flat variant $\sexp=0$ applies no seasonal adjustment and recovers the core
forecaster's scoring density and path-sampling distribution. Assume
uniform initial variant weights and positive mixture density at every observation.
If this variant's single-period excess log loss relative to the best variant in
that period is always bounded by
$\delta$, its mixture weight cannot fall below
$\nvar^{-1}e^{-\delta/(1-\fgt)}$, and the mixture's per-period excess log loss over it is at
most $\log\nvar+\delta/(1-\fgt)$. These discounted bounds require the stated
single-period condition. With undiscounted weights, the cumulative excess log loss
relative to the best fixed variant is at most $\log\nvar$ for every sequence;
discounting trades that horizon-free guarantee
for adaptation over the year (Proposition~\ref{prop:shape-bound}). Both results
characterize the log loss of the scoring densities; the shape layer's effect on control reward is
evaluated empirically.

\subsection{Deterministic lead times}
\label{sec:determ}

With deterministic $\lead$, the order at $\per$ first affects available stock at
$\per+\lead$. Theorem~\ref{thm:zero-option} characterizes the zero-lead-time
case. Positive lead times introduce a pipeline state
(Remark~\ref{app:rem:leadtime}), which we handle with the following
projected-inventory surrogate. Write
$\Short=\dem_{\per+\lead}-\Iproj_{\per+\lead}$ for the shortfall immediately
before that order arrives, where $\Iproj_{\per+\lead}$ is obtained by simulating
the committed pipeline. The $\ord$-th unit earns $\prof$ if
$\Short\geq\ord$. Otherwise it contributes to a surplus of
$\pos{\ord-\Short}$ units. Approximating clearance at predictive mean rate
$\mu$, we assign marginal value
\begin{equation}
m(\ord) = \prof\,\Prob(\Short \ge \ord)
  - \hold\,\E\!\left[
     \left(1 + \dep \frac{\pos{\ord - \Short}}{\mu}\right)
     \ind{\Short < \ord}\right].
\label{eq:marginal}
\end{equation}

The marginal value $m(\ord)$ defines the surrogate objective
$\int_0^{\ord}m(u)\,du$. The fitted multiplier $\dep$ adjusts the mean-rate
clearance approximation. Under the conditions of Proposition~\ref{prop:marginal},
$m$ is non-increasing and $\ord^0=\inf\{\ord\ge0:m(\ord)\le0\}$ is the least
maximizer of this integrated surrogate. When $\dep=0$, this reduces to the
nonnegative $\crat$-fractile of $\Short$. The zero-lead-time theorem motivates
the cost-based construction; the implemented deterministic policy uses the
fitted depletion multiplier, including $\dep=1$ at $\lead=0$.
Appendix~\ref{app:method-determ} gives its parameterization, validation fit, and root search;
$\S$\ref{sec:exp-protocol} describes how its functional form was chosen.

The lost-sales recursion projects stock using sales along each sampled path. A textbook
inventory-position target uses cumulative demand over the protection interval
\citep{zipkin2000foundations}. It therefore treats demand observed at an empty
shelf as if it had reduced inventory position, even though that demand was lost.
For the same sampled path, the textbook shortfall equals our projected shortfall
plus the lost sales inside the interval (Proposition~\ref{prop:overorder}). It is
thus no smaller on any path, and the means differ by expected lost sales. This identity
concerns shortfalls and their means. The order gap depends on the nonlinear
quantile or crossing in \eqref{eq:marginal}.
Appendix~\ref{app:rem:not-mean} gives the associated quantile bound.

\subsection{Stochastic lead times}
\label{sec:stoch}

In the stochastic cells, $\lead_\per\in\{1,2,3,\infty\}$ and an infinite lead
time destroys the order. The realized sequences destroy \deadsyn{} of synthetic
orders and \deadreal{} of real orders. The deterministic rule in
$\S$\ref{sec:determ} assigns an order to one arrival period; stochastic arrivals
couple it to several possible future periods. We evaluate replenishment rules through
rollouts over these supply outcomes. The evaluator's accounting information
identifies whether a positive order survived. Specifically,
$\ord_\per\ind{\lead_\per<\infty}=\intr_{\per+1}-\intr_\per+\arr_\per$; the
survival indicator is therefore known one period later
(Proposition~\ref{prop:supply}). We use these observations to update the survival
probability used in rollouts.

The policy evaluates a small family of complete replenishment rules over the remaining
horizon. Candidate rules use common simulated demand and supply paths, sharing
demand, survival, and finite-lead-time draws to reduce comparison variance. A coarse pass locates promising
parameters, a local grid refines them jointly, and candidate actions are blended
according to estimated value relative to simulation uncertainty. The full search
is specified in Appendix~\ref{app:method}.

The design-time loop selected the capped base-stock family
\begin{equation}
\Fcap:\qquad
\ord=\min\!\left\{\pos{\lev\rate-\IP},\;\capm\rate\right\},
\label{eq:capped}
\end{equation}
where $\rate$ is a local demand forecast and $\lev,\capm$ are searched online.
The position term reacts to low inventory, while the cap prevents a destroyed
order from triggering an arbitrarily large catch-up order. A pure rate rule sets
orders from the demand rate alone, limiting catch-up bursts while leaving its
order unchanged across inventory levels. The
capped family contains the pure position family at $\capm=\lev$ and approaches
the pure rate family as $\lev\to\infty$ (Proposition~\ref{prop:capped}). On a
finite lost-sales instance with order destruction, Appendix~\ref{app:numeric}
provides a numerical certificate of strict reward separation, conditional on
the validity of the estimated slope bounds.

Capped base-stock control is a standard family also used by InventoryBench's
OR baseline. In our implementation, its cap scales with a
local forecast rate, both parameters are selected online for each instance, and
the family was retained by the design-time validation loop. In the small exhaustive
examples in Appendix~\ref{app:numeric}, the
advantage disappears when the lead-time support prevents surviving orders from
arriving in a lump, and is positive in most nearby configurations where such a
lump can occur.

\section{Experiments}
\label{sec:exp}

\subsection{Benchmark, baselines, and protocol}
\label{sec:exp-bench}
\label{sec:exp-protocol}

We evaluate all $1{,}320$ InventoryBench instances \citep{inventorybench}; the two
stochastic cells account for one third of the overall mean. Table~\ref{tab:baselines}
reports the ten public baselines. The OR row is the benchmark authors' capped
base-stock heuristic. The remaining rows use GPT-5 Mini, Gemini 3 Flash, or Grok
4.1 Fast each period, either to choose the order directly (\emph{LLM}), estimate
inputs to the OR rule (\emph{LLM$\to$OR}), or revise the OR recommendation
(\emph{OR$\to$LLM}). We recompute each cell from the released instance-level
results; their weighted means reproduce the published overall scores. SabreAgent
scores \ourscore{} against \pubbest{} for the best public baseline and leads all
six cells. This comparison varies both model choice and placement: SabreAgent
uses DeepSeek-V4-Flash at design time. The same-model online-adjustment control
below measures the effect of adding bounded per-period adjustments by that model
to the complete SabreAgent policy.

\begin{table}[t]
\centering
\scriptsize
\begin{tabular}{@{}lrrrrrrr@{}}
\toprule
& \multicolumn{3}{c}{Synthetic} & \multicolumn{3}{c}{Real} & \\
\cmidrule(lr){2-4}\cmidrule(lr){5-7}
Method & $\lead=0$ & $\lead=4$ & stoch. & $\lead=0$ & $\lead=4$ & stoch.
& Overall \\
\midrule
Gemini 3 Flash, OR$\to$LLM & 0.8391 & 0.5395 & 0.4107 & 0.7439 & 0.3071 & 0.3531 & 0.5380 \\
Grok 4.1 Fast, OR$\to$LLM & 0.8493 & 0.5473 & 0.3574 & 0.6914 & 0.2918 & 0.3034 & 0.5139 \\
Gemini 3 Flash, LLM$\to$OR & 0.8604 & 0.5576 & 0.2434 & 0.7716 & 0.3033 & 0.2375 & 0.5009 \\
Gemini 3 Flash, LLM & 0.8253 & 0.4341 & 0.3865 & 0.7364 & 0.2400 & 0.3119 & 0.4945 \\
Grok 4.1 Fast, LLM$\to$OR & 0.8584 & 0.5555 & 0.2364 & 0.7696 & 0.3085 & 0.1981 & 0.4934 \\
GPT-5 Mini, OR$\to$LLM & 0.8525 & 0.4937 & 0.1643 & 0.7645 & 0.2840 & 0.1819 & 0.4611 \\
GPT-5 Mini, LLM & 0.8371 & 0.4900 & 0.1875 & 0.7646 & 0.2879 & 0.1643 & 0.4597 \\
Grok 4.1 Fast, LLM & 0.8387 & 0.3166 & 0.3959 & 0.7068 & 0.1555 & 0.3073 & 0.4593 \\
GPT-5 Mini, LLM$\to$OR & 0.8452 & 0.5302 & 0.1593 & 0.7319 & 0.2797 & 0.1681 & 0.4578 \\
OR (capped base-stock) & 0.8179 & 0.5359 & 0.1068 & 0.6708 & 0.3017 & 0.2100 & 0.4447 \\
\midrule
\rowcolor{oursbg} SabreAgent (ours) & \textbf{0.8649} & \textbf{0.6685} & \textbf{0.5063} & \textbf{0.8018} & \textbf{0.4447} & \textbf{0.4714} & \textbf{0.6311} \\
\midrule
\textit{Constant level, hindsight-tuned} & \textit{0.8366} & \textit{0.6220} & \textit{0.4653} & \textit{0.7051} & \textit{0.4185} & \textit{0.4563} & \textit{0.5892} \\
\textit{Supply-feasible ceiling} & \textit{1.0000} & \textit{0.9155} & \textit{0.6801} & \textit{1.0000} & \textit{0.8945} & \textit{0.7039} & \textit{0.8656} \\%
 
\bottomrule
\end{tabular}
\caption{Normalized reward by benchmark cell. Bold marks the best deployable
policy in each column. Public rows are recomputed from released instance-level
results; SabreAgent is scored by the official evaluator. Italic rows use hindsight
and provide diagnostic reference scores.}
\label{tab:baselines}
\end{table}

Automated selection uses fresh-seed instances generated from the documented
synthetic families and only the five initial demands available to a policy. The
harness cannot access benchmark instances, and the final configuration is chosen
from validation scores (Appendix~\ref{app:expsupp}). Reusing this validation set
across rounds incurs adaptive-selection bias within the validation process
\citep{dwork2015reusable,blum2015ladder}. Development also includes two indirect
test-informed choices. We chose
the functional form of $\dep$ after observing a systematic cell-level mismatch,
then fit its coefficients on validation. The proposer context contains four
aggregate test diagnostics that identify stochastic supply as the weak regime,
although it contains no test score for any proposed candidate. The reported test
result therefore reflects validation-based candidate selection within a
development process informed by these two test-diagnostic channels.

\begin{figure}[t]
\centering
\includegraphics[width=0.72\textwidth]{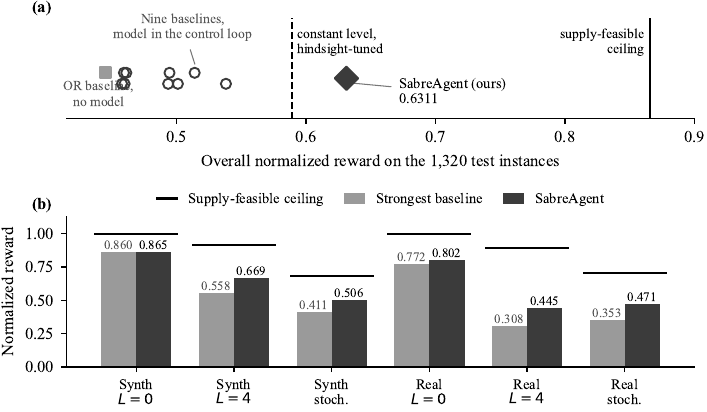}
\caption{Performance relative to two hindsight diagnostics. (a) Public baselines
against a constant order-up-to level chosen per instance (dashed) and a
clairvoyant supply-feasible upper bound (solid, \ceiling{}). (b) Cell-wise results;
the upper bound equals $1$ only when lead time is zero, so equal vertical gaps do
not represent equal attainable headroom.}
\label{fig:standing}
\end{figure}

\subsection{Supply-imposed headroom}
\label{sec:exp-ceiling}

$\Score=1$ sets the normalization reference. Attainable reward also depends on
the supply process. We
construct a clairvoyant plan that knows demand and realized lead times but remains
subject to feasible arrival times and order destruction. Because orders are free
and uncapacitated, each unit of demand can be assigned independently to its latest
feasible arrival whenever its profit exceeds the intervening holding cost
(Appendix~\ref{app:expsupp}). The resulting ex-post upper bound is \ceiling{}
overall, exactly $1$ in both zero-lead-time cells, and \ceilsynstoch{} and
\ceilrealstoch{} in the stochastic cells. Much of their gap to $1$ is therefore
imposed by the supply constraints.

With $\dep=1$ and realized future demand supplied as a point-mass forecast, the
projected base-stock rule matches this bound in all four deterministic cells.
This observation concerns the fixed rule with exact demand inputs;
positive-lead-time optimality under a demand distribution is a separate question.
Appendix~\ref{app:expsupp} specifies the oracle configuration.

A second diagnostic chooses one constant order-up-to level per instance with
hindsight and never revises it. It scores \hindsight{}, above all ten public
baselines overall. The contrast is strongest at lead time four: the constant
scores \hsleadsyn{} and \hsleadreal{} on synthetic and real data, compared with
\bestleadsyn{} and \bestleadreal{} for the strongest public baseline in each
cell. At zero lead time, eight of ten baselines exceed the constant in each cell,
so this observation is specific to settings with a protection interval.

\begin{figure}[t]
\centering
\includegraphics[width=0.66\textwidth]{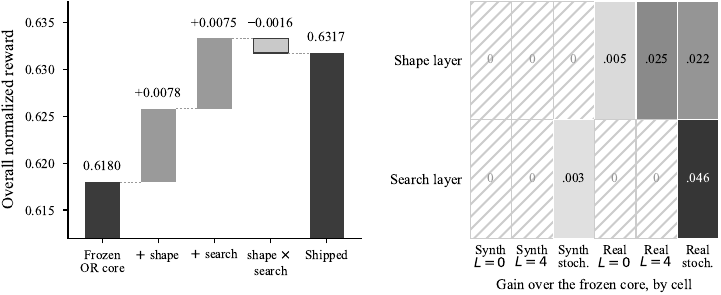}

\vspace{3pt}
\scriptsize
\setlength{\tabcolsep}{3pt}
\begin{tabular}{@{}lrrrrrrr@{}}
\toprule
& \multicolumn{3}{c}{Synthetic} & \multicolumn{3}{c}{Real} & \\
\cmidrule(lr){2-4}\cmidrule(lr){5-7}
Configuration & $\lead=0$ & $\lead=4$ & stoch. & $\lead=0$ & $\lead=4$ & stoch.
& Overall \\
\midrule
\input{tables/ablation.tex} 
\bottomrule
\end{tabular}
\normalsize
\caption{Factorial ablation over six cells and three seeds on the batched path
used during design-time search. Moving to the official one-instance-at-a-time
path changes the shipped score by \execgap{}; values in this figure should
therefore be compared within this execution path. Table~\ref{tab:baselines}
reports the official path.
The plot decomposes the overall gain, including interaction; the table reports
the frozen core in absolute terms and all other rows as deltas from it.}
\label{fig:ablation}
\end{figure}

\subsection{Component contributions and online-adjustment control}
\label{sec:exp-abl}
\label{sec:exp-null}

The frozen OR core reaches \abcore{} on the batched ablation path, above every
published baseline's overall score. The two model-authored components add
\abauthored{} on that same path, which is \abshare{} of the gap between the final
score and the strongest public baseline. The seasonal layer acts only on the
three real-data cells, where its paired mean effect is \abshapewhere{}. The search
layer acts only on the two stochastic-lead-time cells, where its paired mean
effect is \absearchwhere{} (Table~\ref{tab:app-paired}). Averaging in the zero
effects from the remaining cells reduces their overall deltas. The additive
decomposition uses the batched path throughout; official one-instance-at-a-time
scores are reported separately.

Figure~\ref{fig:ablation} crosses three binary choices: seasonal shapes on or off,
the initial or selected stochastic policy family, and the core forecaster or the
best alternative found during development. The search contrast includes both
research passes: local-rate normalization from the first and the cap from the
second. Table~\ref{tab:app-validation} separates them on validation. This
factorial interpretation is valid because accepted changes affect separate
components;
Appendix~\ref{app:method-compose} describes the checks used when composing them.

The seasonal and search layers overlap only in the real stochastic cell;
their interaction contributes \abinter{} to the overall weighted mean. The search
family and alternative forecaster were selected from validation candidates.
Their test contrasts also reflect the test-informed development choices described
in \S\ref{sec:exp-protocol}. The seasonal shapes were not selected from test
outcomes. Paired bootstrap intervals exclude zero for the seasonal layer, the
search layer, and their combination, both overall and within the cells where each
operates. The interval for the pooled forecaster includes zero
(Table~\ref{tab:app-paired}).

For each bootstrap replicate, we sample test instances with replacement and
retain the complete row of scores for all arms, preserving their pairing.
Across the $1{,}320$ test instances, the combined layers yield a single-seed
paired mean gain of $+0.0139$ over the frozen core, with a 95\% percentile
interval of $[+0.0115,\, +0.0164]$ (Table~\ref{tab:app-paired}).
The combined layers change scores on 62\% of instances, with the ten largest
movers contributing 14\% of total absolute movement.
This estimate uses the instance-level scores for one seed; the
\abauthored{} gain in Figure~\ref{fig:ablation} averages three seeds.
The final configuration's across-seed standard deviation of
\abseedsd{} measures rollout Monte Carlo noise on fixed data.
Together, these summaries distinguish rollout variability from sensitivity
to the composition of the fixed test set, which is the scope of the paired
intervals.

Most proposed changes are rejected. Across two research passes, one of fifty-two
candidates contributes to the final search component, and that candidate includes
the useful part of the first pass. No proposal targeting deterministic cells
survives validation (Appendix~\ref{app:expsupp}). This screening record shows
which proposed changes met the validation criterion before deployment.

\begin{table}[t]
\centering
\scriptsize
\begin{tabular}{@{}lrrrrrr@{}}
\toprule
& \multicolumn{3}{c}{Synthetic} & \multicolumn{3}{c}{Real} \\
\cmidrule(lr){2-4}\cmidrule(lr){5-7}
A model returned to the per-period loop & $\lead=0$ & $\lead=4$ & stoch.
& $\lead=0$ & $\lead=4$ & stoch. \\
\midrule
change against the shipped policy & $-$0.0010 & $-$0.0020 & $+$0.0057 & $-$0.0008 & $-$0.0299 & $+$0.0104 \\
\quad standard error & 0.0008 & 0.0081 & 0.0059 & 0.0009 & 0.0111 & 0.0046 \\
 
\bottomrule
\end{tabular}
\caption{Same-model online-adjustment control. DeepSeek-V4-Flash adds
per-period order adjustments, using a bounded multiplier, to the complete
SabreAgent policy. Entries are changes in normalized
reward on a stratified \ctrln{}-instance test subset; standard errors are shown in
the second row.}
\label{tab:control}
\end{table}

On the stratified subset in Table~\ref{tab:control}, we add
bounded per-period order adjustments by DeepSeek-V4-Flash to the complete
SabreAgent policy. All four deterministic-cell estimates are negative,
but only the real lead-time-four estimate exceeds two standard errors in
magnitude. The two stochastic estimates are positive. Within this subset, the real
stochastic cell is the only cell with a positive estimate exceeding two standard
errors. These results apply to DeepSeek-V4-Flash, this bounded adjustment rule, and
this test subset. Broader conclusions about deterministic supply require further
evaluation.

The highest-scoring validation configuration pools information across instances,
which violates the official one-policy-per-instance protocol. We therefore do
not submit it. Its mean gain is \abrpooled{} over three seeds, while the
single-seed paired interval includes zero (Table~\ref{tab:app-paired}). The
submitted configuration preserves the required per-instance information boundary.

\paragraph{Budget and reproducibility.}
\label{sec:exp-repro}
Producing the shipped artifacts used \dtcalls{} calls to DeepSeek-V4-Flash,
mostly three shape samples for each of the \shapeitems{} real products. Calling a
model once per period over the full test set would require \ddcalls{} calls, about
\dtratio{} times as many (Figure~\ref{fig:app-budget}). During operation, a fixed seed
makes the policy reproducible; it is rerun end to end and scored with the official
evaluator as detailed in Appendix~\ref{app:expsupp}.

\section{Conclusion}
\label{sec:conclusion}

On InventoryBench, most of SabreAgent's performance gain comes from an OR policy that uses the
benchmark's free replenishment, zero salvage, observable demand, and explicit
order destruction. Product-specific seasonal priors supply timing cues in the
real-data cells, and validation selects a capped base-stock family for lossy
supply. These components provide additional gains while statistical forecasting
and inventory optimization determine orders during operation. The final policy
leads every benchmark cell with zero runtime language-model calls.

The zero-lead-time theorem identifies conditions for exact myopic optimality;
positive purchase or setup costs change the value of carried inventory. The
controller for positive deterministic lead times uses a surrogate objective, and the seasonal
construction gives predictive-log-loss bounds under the conditions in
$\S$\ref{sec:shape}. Inventory reward is measured in the benchmark experiments.
The online-adjustment control evaluates bounded order changes by one model on
\ctrln{} instances. Matched comparisons across models and larger samples would
clarify when these additional runtime adjustments improve the complete policy.

\bibliography{sabreagent}
\bibliographystyle{iclr2025_conference}

\appendix
\section{Related work}
\label{sec:related}

\paragraph{Lost-sales inventory control.}
Positive-lead-time lost-sales systems generally require the full pipeline state,
and ordinary base-stock policies need not be optimal
\citep{karlin1958inventory,zipkin2008structure,zipkin2008old,bijvank2011lost}.
Different asymptotic regimes favor order-up-to policies when lost-sales penalties
grow \citep{huh2009asymptotic}, constant-order policies when lead times grow
\citep{goldberg2016asymptotic,xin2016optimality}, and capped base-stock policies
between those extremes \citep{xin2021understanding}. InventoryBench already uses
a capped base-stock baseline \citep{inventorybench}. Building on this established
family, we connect the cap to catch-up bursts caused by
destroyed orders, scale it by a local demand rate, and select both parameters
online. Theorem~\ref{thm:zero-option} concerns a different, degenerate regime:
without purchase or setup costs there is no speculative or batching motive, and
instantaneous replenishment removes protection-time value as well.

\paragraph{Data-driven inventory policies and forecasting.}
Data-driven approaches replace a known distribution with samples
\citep{levi2007provably}, estimate decisions from features
\citep{ban2019bigdata}, or train policies on simulated and operational data
\citep{gijsbrechts2022can,oroojlooyjadid2022deep,qi2023practical}. InventoryBench
instead provides five initial observations and no cross-series training set, which
favors a mixture of simple models over a flexible learned forecaster. Discounted
predictive-likelihood weighting is standard
\citep{hoeting1999bayesian,raftery2010online}, as are the short-sample forecasting
members we combine \citep{makridakis2020m4,hyndman2008automatic}. The technical
detail specific to our mixture is that members fit in different transform spaces
must be scored as densities of the same raw observation; otherwise the omitted
Jacobian systematically favors the log-space members.

\paragraph{Language models for search and decision making.}
Verifier-guided systems use a language model to propose an artifact and an
external evaluator to retain or reject it
\citep{romera2024mathematical,liu2024evolution,ye2024reevo,novikov2025alphaevolve,lu2026automation}.
The closest pattern is language-model variation inside evolutionary search over
heuristics \citep{hemberg2024evolving}. In operations research, prior work has
translated natural language into formal models checked by a
solver \citep{ramamonjison2023nl4opt,ahmaditeshnizi2024optimus,xiao2024chainofexperts}.
SabreAgent applies this established search pattern to policy-family design.
This choice also reflects evidence that
language models remain unreliable as planners or self-correctors
\citep{valmeekam2023planbench,huang2024selfcorrect,kambhampati2024position} and can
show systematic biases in newsvendor decisions \citep{liu2025newsvendor}.

\paragraph{Text-derived priors for cold starts.}
Language models have also been adapted for time-series forecasting
\citep{gruver2023llmtime,jin2024timellm,zhang2024llmforecaster}, although some
reported gains survive removal of the language component \citep{tan2024language}.
Our seasonal layer uses product descriptions to elicit a prior over relative
demand timing. Related work has used language models to infer priors from
metadata, mix parametric priors over prompts, or produce
predictive beliefs from textual descriptions
\citep{garthwaite2005statistical,choi2022lmpriors,capstick2025autoelicit,requeima2024llmprocesses,selby2025experts}.
The resulting seasonal variant always competes with a statistical variant without
seasonal adjustment. This design addresses the benchmark's cold start. When
retailer data on comparable products are available, clustering, attribute
regression, and life-cycle models provide more directly grounded signals
\citep{thomassey2006hybrid,ferreira2016analytics,hu2019forecasting,vansteenbergen2020forecasting,fildes2022retail}.

\section{Proofs}
\label{app:proofs}

Theorem~\ref{thm:zero-option} is the one result the main text states formally,
and it is restated here before it is proved, so the restatement carries a fresh
number. The six results after it are discussed there in prose and stated formally
only here, one to a subsection, in the order the main text raises them. Standing
conventions: demand is nonnegative and integrable, $\prof > 0$ and $\hold > 0$,
$\crat = \prof/(\prof+\hold) \in (0,1)$, and $\pos{x} = \max\{x,0\}$. For a
random variable $Z$ we write $F_Z(z) = \Prob(Z \le z)$ and $F_Z(z^-) = \Prob(Z <
z)$, which differ exactly at the atoms of $Z$; this distinction determines the
one-sided conditions in Proposition~\ref{app:prop:marginal} and is kept throughout.

Two distinctions organize the results. The exact-optimality clause of
Theorem~\ref{thm:zero-option} rests on Assumption~\ref{app:as:levels};
\S\ref{app:sub:thm1} states this condition and gives a counterexample when it
fails. Proposition~\ref{prop:overorder} gives an exact identity for shortfalls
and their means, together with bounds for the corresponding quantile gaps
(\S\ref{app:sub:prop4}).

\subsection{Theorem~\ref{thm:zero-option}: zero option value at \texorpdfstring{$\lead = 0$}{L=0}}
\label{app:sub:thm1}

\begin{assumption}[the zero-option setting, restated from \S\ref{sec:prelim-cost}]
\label{app:as:zero-option}
The per-period reward is $\prof\,\sold_\per - \hold\,\onhend_\per$ with
$\prof, \hold > 0$ and with no purchase, fixed-ordering or salvage term; unmet
demand is lost; orders are unbounded above and arrive in the period in which
they are placed, $\lead = 0$; the horizon $\Tend$ is finite; and demand is
exogenous, nonnegative, integrable and independent across periods with known
marginals $F_\per$.
\end{assumption}

\paragraph{Setting.} Fix an instance satisfying
Assumption~\ref{app:as:zero-option}.
With $\lead = 0$ an order placed at $\per$ lands inside period $\per$, so the
decision at $\per$ is equivalent to choosing the post-order level $\base \ge
\onh_\per$. Writing $\dem_\per$ for the period's demand, the period earns
$\prof \min(\dem_\per, \base) - \hold \pos{\base - \dem_\per}$ and carries
$\pos{\base - \dem_\per}$ into $\per+1$, since unmet demand is lost and holding
is charged on ending inventory. Define the one-period expected reward and the
dynamic program
\begin{equation}
  G_\per(\base) = \E\!\left[\prof \min(\dem_\per,\base)
      - \hold \pos{\base - \dem_\per}\right],
  \qquad
  V_\per(\onh) = \sup_{\base \ge \onh}
      \Big\{ G_\per(\base) + \E\big[V_{\per+1}(\pos{\base - \dem_\per})\big]\Big\},
  \label{eq:app-dp}
\end{equation}
with $V_{\Tend+1} \equiv 0$. We write $J_\per(\base)$ for the bracketed
objective in \eqref{eq:app-dp} and $C_\per(\base) = \E[V_{\per+1}(\pos{\base -
\dem_\per})]$ for its continuation term, so $J_\per = G_\per + C_\per$. A
\emph{$\crat$-fractile} of $F_\per$ is any $\base^*$ with $F_\per(\base^{*-}) \le
\crat \le F_\per(\base^*)$.

\begin{lemma}[the one-period function]
\label{app:lem:newsvendor}
$G_\per$ is real-valued and concave on $\R$, with one-sided derivatives
$G_\per^{\prime+}(\base) = \prof - (\prof+\hold)F_\per(\base)$ and
$G_\per^{\prime-}(\base) = \prof - (\prof+\hold)F_\per(\base^-)$. Its set of
maximizers over $\R$ is exactly the set of $\crat$-fractiles of $F_\per$, every
such point is nonnegative, $G_\per$ is non-decreasing on
$(-\infty,\base^*_\per]$ and non-increasing on $[\base^*_\per,\infty)$ for any
$\crat$-fractile $\base^*_\per$.
\end{lemma}

\begin{proof}
For each fixed $\dem \ge 0$ the map $\base \mapsto \min(\dem,\base)$ is concave
and $\base \mapsto \pos{\base-\dem}$ is convex, so the integrand is concave in
$\base$; concavity is preserved by the expectation, which is finite because
$|\min(\dem_\per,\base)| \le |\base|$ and $0 \le \pos{\base-\dem_\per} \le
\pos{\base}$, so $G_\per$ is finite at every $\base$ with no integrability
hypothesis at all. Integrability of $\dem_\per$ is needed for $V_\per$, not here.

For the right derivative, $\min(\dem,\base+\varepsilon) - \min(\dem,\base) =
\min(\pos{\dem-\base},\varepsilon)$, so the difference quotient
$\varepsilon^{-1}\min(\pos{\dem-\base},\varepsilon)$ increases to
$\ind{\dem > \base}$ as $\varepsilon \downarrow 0$ and monotone convergence
gives $\Prob(\dem_\per > \base)$. Likewise $\pos{\base+\varepsilon-\dem} -
\pos{\base-\dem} = \varepsilon\ind{\dem \le \base} + (\base+\varepsilon-\dem)
\ind{\base < \dem < \base+\varepsilon}$, whose difference quotient converges to
$\ind{\dem \le \base}$ boundedly, giving $F_\per(\base)$. Hence
$G_\per^{\prime+}(\base) = \prof(1 - F_\per(\base)) - \hold F_\per(\base) =
\prof - (\prof+\hold)F_\per(\base)$, and the same computation from the left
gives $F_\per(\base^-)$.

By concavity, $\base^*$ maximizes $G_\per$ if and only if
$G_\per^{\prime-}(\base^*) \ge 0 \ge G_\per^{\prime+}(\base^*)$, that is
$F_\per(\base^{*-}) \le \crat \le F_\per(\base^*)$, which is the definition of a
$\crat$-fractile. Since $\dem_\per \ge 0$ we have $F_\per(\base) = 0$ for
$\base < 0$, so $G_\per^{\prime+}(\base) = \prof > 0$ there and no maximizer is
negative. Monotonicity on either side of $\base^*_\per$ is concavity together
with the sign of the one-sided derivatives.
\end{proof}

Lemma~\ref{app:lem:newsvendor} is where the zero purchase cost enters. A unit
purchase cost $\cbuy > 0$ would replace the objective in \eqref{eq:app-dp} by
$J_\per(\base) - \cbuy(\base - \onh_\per)$, in which the entering inventory
appears as a linear credit, not only through the constraint
$\base \ge \onh_\per$; a carried unit would then save $\cbuy$ at every later
period, which is precisely
the option value the theorem denies.

For the exact-optimality clause of Theorem~\ref{thm:zero-option}, we state
the non-decreasing-fractile condition as a separate assumption.

\begin{assumption}[monotone myopic levels]
\label{app:as:levels}
The $\crat$-fractiles can be selected non-decreasing in time: there exist
$\crat$-fractiles $\base^*_1 \le \base^*_2 \le \dots \le \base^*_{\Tend}$ of
$F_1,\dots,F_{\Tend}$.
\end{assumption}

Assumption~\ref{app:as:levels} holds whenever the demands are identically
distributed, in which case all the levels coincide, and more generally whenever
$\dem_\per$ is stochastically non-decreasing in $\per$, since then
$F_{\per+1} \le F_\per$ pointwise and the smallest $\crat$-fractiles are already
non-decreasing. Remark~\ref{app:rem:levels-needed} and
Example~\ref{app:ex:counter} below show that it cannot be dropped.

\begin{theorem}[Exact-optimality clause of Theorem~\ref{thm:zero-option}]
\label{app:thm:zero-option}
Let Assumption~\ref{app:as:zero-option} hold, let the demands be independent with
known laws $F_\per$ and finite means, and let Assumption~\ref{app:as:levels}
hold with the selection $\base^*_1 \le \dots \le \base^*_{\Tend}$. Then for
every $\per \le \Tend$ the value function $V_\per$ of \eqref{eq:app-dp} is
constant on $[0,\base^*_\per]$ and non-increasing on $[0,\infty)$, the supremum
in \eqref{eq:app-dp} is attained at $\base = \max\{\onh, \base^*_\per\}$, and
the policy $\base_\per = \max\{\onh_\per, \base^*_\per\}$ is optimal.
Consequently a unit carried from $\per-1$ into $\per$ has zero marginal value
whenever $\onh_\per \le \base^*_\per$, and the myopic newsvendor solution is
exactly optimal, not asymptotically or approximately so.
\end{theorem}

This restatement is the conditional half of Theorem~\ref{thm:zero-option}. The
unconditional half, that $V_\per$ is non-increasing and that no post-order level
above $\max\{\onh,\base^*_\per\}$ is strictly better, needs neither
Assumption~\ref{app:as:levels} nor the induction below and is separated out as
Proposition~\ref{app:prop:no-levels}. It is the half the design of
\S\ref{sec:determ} rests on, so the two are proved apart to keep visible which
one carries which hypothesis.

The benchmark fixes each instance's demand path in advance and takes the
expectation over the policy's own randomization alone
(\S\ref{sec:prelim-problem}), while the theorem treats demand as random with
known marginals $F_\per$. The theorem describes the policy's stochastic
planning model. Taking
$F_\per = \delta_{\dem_\per}$ degenerate recovers the latter, and the myopic
policy then becomes the clairvoyant rule that \S\ref{sec:exp-ceiling}
measures.

\begin{proof}
Backward induction on $\per$, with induction hypothesis the conjunction

\smallskip
\noindent
\textbf{(H1)} $V_{\per+1}$ is constant on $[0, \base^*_{\per+1}]$, and
\textbf{(H2)} $V_{\per+1}$ is non-increasing on $[0,\infty)$,

\smallskip
\noindent
where for $\per+1 = \Tend+1$ we read $\base^*_{\Tend+1} = +\infty$.

\emph{Base case.} At $\per = \Tend$ the continuation is $V_{\Tend+1} \equiv 0$,
which is constant on all of $[0,\infty)$ and non-increasing, so (H1) and (H2)
hold. This is the only place the zero salvage value is used, and it is used
decisively: a salvage credit $\varsigma > 0$ would give
$V_{\Tend+1}(\onh) = \varsigma\,\onh$, which is strictly increasing and
violates both halves of the hypothesis at once.

\emph{Inductive step.} Assume (H1) and (H2) at $\per+1$ and fix $\per \le
\Tend$.

\emph{Step 1: the continuation is non-increasing.} For every $\dem \ge 0$ the
map $\base \mapsto \pos{\base-\dem}$ is non-decreasing, so by (H2) the map
$\base \mapsto V_{\per+1}(\pos{\base-\dem})$ is non-increasing; taking
expectations, $C_\per$ is non-increasing on $[0,\infty)$. Independence of
$\dem_\per$ from the past is what makes $C_\per$ a function of $\base$ alone.

\emph{Step 2, the key step: the continuation is flat below the myopic level.}
Let $0 \le \base \le \base^*_\per$. Then $\pos{\base - \dem} \le \base \le
\base^*_\per \le \base^*_{\per+1}$ for every $\dem \ge 0$, where the last
inequality is Assumption~\ref{app:as:levels}. So the argument of $V_{\per+1}$
lies in $[0,\base^*_{\per+1}]$ pathwise, and (H1) gives
$V_{\per+1}(\pos{\base-\dem_\per}) = V_{\per+1}(0)$ almost surely. Hence
$C_\per(\base) = V_{\per+1}(0)$ for all $\base \in [0,\base^*_\per]$: on that
range the continuation contributes an additive constant, and the first-order
condition for $J_\per$ is the first-order condition for $G_\per$ alone. This is
the mechanism the theorem asserts. It is available because ordering is free and
instantaneous, so any unit not carried forward can be replaced at $\per+1$ at
no cost, which is exactly what makes $V_{\per+1}$ flat below its own target.

\emph{Step 3: $\base^*_\per$ maximizes $J_\per$.} On $[0,\base^*_\per]$ we have
$J_\per = G_\per + V_{\per+1}(0)$, which by Lemma~\ref{app:lem:newsvendor} is
non-decreasing and is maximized over that interval at $\base^*_\per$. For
$\base > \base^*_\per$, Lemma~\ref{app:lem:newsvendor} gives $G_\per(\base) \le
G_\per(\base^*_\per)$ and Step~1 gives $C_\per(\base) \le C_\per(\base^*_\per)$,
so $J_\per(\base) \le J_\per(\base^*_\per)$. Therefore $\base^*_\per$
maximizes $J_\per$ over $[0,\infty)$, and $J_\per$ is non-increasing on
$[\base^*_\per,\infty)$ because both of its terms are.

\emph{Step 4: re-establishing the hypothesis at $\per$.} For $\onh \le
\base^*_\per$ the unconstrained maximizer $\base^*_\per$ is feasible, so
$V_\per(\onh) = J_\per(\base^*_\per)$. For $\onh > \base^*_\per$ the objective
is non-increasing on the feasible set $[\onh,\infty)$ by Step~3, so the
supremum is attained at $\base = \onh$ and $V_\per(\onh) = J_\per(\onh)$.
Both cases are summarized by $V_\per(\onh) = J_\per(\max\{\onh,
\base^*_\per\})$, which is the claimed optimal post-order level. It follows
that $V_\per$ is constant on $[0,\base^*_\per]$, giving (H1) at $\per$, and
that $V_\per$ is non-increasing on $[0,\infty)$, since it is constant up to
$\base^*_\per$ and equals $J_\per$, which is non-increasing, above it. This
gives (H2) at $\per$ and closes the induction.

Zero marginal value of carried stock is (H1) read as a difference: for
$0 \le \onh' \le \onh \le \base^*_\per$, $V_\per(\onh) - V_\per(\onh') = 0$.
Optimality of $\base_\per = \max\{\onh_\per,\base^*_\per\}$ is Step~4 applied
at every period. That a Markov policy attaining the per-period maximum in
\eqref{eq:app-dp} is optimal among all history-dependent policies is the
verification theorem for a finite-horizon problem with bounded-below rewards
and a one-dimensional Borel state \citep{bertsekas1978stochastic}; the
expectation $\E[V_{\per+1}(\pos{\base-\dem_\per})]$ is defined because
$V_{\per+1}$ is monotone and hence Borel.
\end{proof}

Zero fixed ordering cost is also part of the theorem's setting. With a setup
charge $K > 0$ the period
objective becomes $J_\per(\base) - K\ind{\base > \onh}$, which is not concave
in $\base$, so Lemma~\ref{app:lem:newsvendor} no longer identifies the
maximizer and the optimal policy is of $(s,S)$ form, no longer base-stock
\citep{scarf1960optimality}. Under $K > 0$ it is optimal not to order at all
for entering inventories somewhat below the target, so $V_\per$ is not flat on
$[0,\base^*_\per]$ and the conclusion fails even though the mechanism of
Step~2 is untouched.

The effect of decreasing myopic levels depends on the reachable leftover inventory.
A high-demand period preceding a low one is not by itself enough to break exact
optimality. Because ordering is free, $V_{\per+1}$ is flat on
$[0,\base^*_{\per+1}]$ by (H1), so carrying surplus into $\per+1$ costs nothing
while the leftover stays under the next level, the policy having been going to
order up to it in any case. What breaking exact optimality needs is that the
drop be \emph{reachable}, $\Prob(\pos{\base^*_\per - \dem_\per} >
\base^*_{\per+1}) > 0$: only then does a unit ordered at the fractile face a
period in which no order would have been placed for it. Reachability is
necessary and not sufficient, since the holding it exposes the unit to must also
outweigh the current period's marginal value. Exact backward induction over a
bounded integer state exhibits all three cases. With $\base^*$ falling from $17$ to
$11$ but demand never below $12$,
so that the leftover cannot exceed $11$, the myopic level stays optimal at every
state despite the drop; with a reach probability of $0.012$ it is still optimal;
and with $\base^*$ falling from $16$ to $1$ at a reach probability of $0.714$ the
optimum drops six units below the fractile.

\begin{corollary}[integer orders]
\label{app:cor:integer}
Assume in addition that $\dem_\per$ is integer valued and that the post-order
level is restricted to the integers, as the benchmark's flooring of order
quantities enforces. Then $G_\per(\base+1) - G_\per(\base) = \prof -
(\prof+\hold)F_\per(\base)$ for integer $\base$, this difference is
non-increasing in $\base$, and the smallest integer $\base$ with
$F_\per(\base) \ge \crat$ maximizes $G_\per$ over the integers. Under the
hypotheses of Theorem~\ref{app:thm:zero-option} the optimal integer post-order
level is $\max\{\onh_\per, \base^{*,\mathbb{Z}}_\per\}$, where
$\base^{*,\mathbb{Z}}_\per = \min\{k \in \mathbb{Z}: F_\per(k) \ge \crat\}$.
\end{corollary}

\begin{proof}
For integer valued $\dem$ and integer $\base$, $\min(\dem,\base+1) -
\min(\dem,\base) = \ind{\dem \ge \base+1}$ and $\pos{\base+1-\dem} -
\pos{\base-\dem} = \ind{\dem \le \base}$. Taking expectations,
$G_\per(\base+1) - G_\per(\base) = \prof\,\Prob(\dem_\per > \base) -
\hold\,\Prob(\dem_\per \le \base) = \prof - (\prof+\hold)F_\per(\base)$, which
is non-increasing in $\base$ because $F_\per$ is non-decreasing. A function on
$\mathbb{Z}$ with non-increasing forward differences is maximized at the first
$\base$ where the forward difference turns non-positive, that is at the
smallest integer with $F_\per(\base) \ge \crat$.

Such an integer is a $\crat$-fractile of $F_\per$, but
Assumption~\ref{app:as:levels} is a property of a selection and the integer
selection need not inherit it: with $\dem_1$ uniform on $\{2,5\}$ and $\dem_2$
uniform on $\{1,4\}$ at $\crat = 1/2$, the monotone selection $(2,2)$ exists
while $\base^{*,\mathbb{Z}} = (2,1)$ decreases. The theorem therefore cannot be
invoked with that selection, and what carries the induction on the lattice is
instead that $\base^{*,\mathbb{Z}}_\per \le \base^*_\per$, being the smallest
such integer, together with the fact that the lattice value function is flat on
$[0,\base^*_{\per+1}] \cap \mathbb{Z}$ and not merely up to
$\base^{*,\mathbb{Z}}_{\per+1}$: either the fractile is unique, in which case
the two levels coincide, or $F_{\per+1} \equiv \crat$ between them, in which
case the forward difference $\prof - (\prof+\hold)F_{\per+1}$ vanishes there
and $G_{\per+1}$ is flat across the gap. Step~2 then runs with
$\pos{\base-\dem} \le \base \le \base^{*,\mathbb{Z}}_\per \le \base^*_\per
\le \base^*_{\per+1}$, and Steps~1, 3 and 4 use only monotonicity of $C_\per$
and of $F_\per$, which the restriction to the lattice preserves.
\end{proof}

The next result gives the upper-bound property under the zero-option setting
alone; the following example identifies the role of Assumption~\ref{app:as:levels}
in exact optimality.

\begin{proposition}[what survives without monotone myopic levels]
\label{app:prop:no-levels}
Let Assumption~\ref{app:as:zero-option} hold with independent, integrable demands,
and let $\base^*_\per$ be any $\crat$-fractile of $F_\per$. Then for every $\per$
and every $\onh \ge 0$,
\begin{equation}
  V_\per(\onh)
  \;=\; \sup_{\onh \le \base \le \max\{\onh,\base^*_\per\}} J_\per(\base),
  \label{eq:app-upper-bound}
\end{equation}
and $V_\per$ is non-increasing on $[0,\infty)$. In words, ordering above the
myopic fractile never strictly pays, so some optimal post-order level lies at or
below $\max\{\onh,\base^*_\per\}$: the myopic level bounds the optimum from
above and never from below. When the fractile is not unique a higher level can
tie, which is why the statement is \eqref{eq:app-upper-bound} and not a claim
that every optimal level is at or below it.
\end{proposition}

\begin{proof}
That $V_\per$ is non-increasing is immediate from \eqref{eq:app-dp}, since the
feasible sets $\{\base \ge \onh\}$ shrink as $\onh$ grows, and it needs no
induction. Given this at $\per+1$, Step~1 of the proof of
Theorem~\ref{app:thm:zero-option} applies verbatim and $C_\per$ is
non-increasing; with Lemma~\ref{app:lem:newsvendor} this makes $J_\per$
non-increasing on $[\base^*_\per,\infty)$, so no $\base > \max\{\onh,
\base^*_\per\}$ can improve on $\max\{\onh,\base^*_\per\}$, which is
\eqref{eq:app-upper-bound}.
\end{proof}

\begin{remark}
\label{app:rem:levels-needed}
Assumption~\ref{app:as:levels} is used only in Step~2, and only to guarantee
that the reachable states $\pos{\base-\dem}$ for $\base \le \base^*_\per$ stay
inside the flat region $[0,\base^*_{\per+1}]$ of $V_{\per+1}$. When the myopic
levels decrease over time this can fail, and where it does it fails in a
direction: stock left over from $\per$ can land above the next period's target,
where it is charged holding without displacing any future order, so the true
optimum orders strictly less than the myopic fractile. A decrease is not by
itself enough, since Step~2 needs only $\pos{\base^*_\per -
\operatorname{essinf}\dem_\per} \le \base^*_{\per+1}$, which is strictly weaker
than Assumption~\ref{app:as:levels} and can hold while the levels fall; the
discussion after Theorem~\ref{app:thm:zero-option} exhibits both sides
numerically. The failure direction is the same as the
inequality of Proposition~\ref{app:prop:no-levels}. The depletion term $\dep$ of
Proposition~\ref{app:prop:marginal} acts in the same direction; $\dep$ remains
an empirical correction fitted for a separate effect.
\end{remark}

\begin{example}[the hypothesis cannot be dropped]
\label{app:ex:counter}
Take $\Tend = 2$, $\prof = \hold = 1$ so $\crat = 1/2$, demands independent with
$\dem_1 \sim \mathrm{Uniform}[0,100]$ and $\dem_2 \equiv 0$. Every hypothesis
of Assumption~\ref{app:as:zero-option} holds, and $\base^*_1 = 50$ while
$\base^*_2 = 0$, so Assumption~\ref{app:as:levels} fails. In period $2$ the
reward of a post-order level $\base \ge \onh$ is $-\base$, so
$V_2(\onh) = -\onh$. Hence
\[
  J_1(\base)
  = \E[\min(\dem_1,\base)] - \E\pos{\base-\dem_1}
    - \E\pos{\base-\dem_1}
  = \E[\min(\dem_1,\base)] - 2\,\E\pos{\base-\dem_1},
\]
and for $\base \in [0,100]$ this is $\base - 3\base^2/200$, maximized at
$\base = 100/3$ with value $50/3 \approx 16.67$, whereas the myopic level
$\base^*_1 = 50$ yields $12.5$. The myopic policy is strictly suboptimal, and
$V_1(\onh) = J_1(\max\{\onh, 100/3\})$ is strictly decreasing on
$(100/3, 50]$, so $V_1$ is not constant on $[0,\base^*_1]$ either. Both halves
of the conclusion of Theorem~\ref{app:thm:zero-option} fail, while
\eqref{eq:app-upper-bound} holds with room to spare.
\end{example}

\begin{remark}[the pipeline state at $\lead > 0$]
\label{app:rem:leadtime}
The zero-lead-time assumption is essential to
Theorem~\ref{app:thm:zero-option}. At $\lead > 0$, the order placed
at $\per$ lands at $\per+\lead$, so the inventory it creates cannot be topped
up before it arrives and the states reachable at $\per+1,\dots,\per+\lead$
depend on $\base$ through the entire pipeline vector rather than a scalar.
Step~2 is the step that fails: the continuation is no longer flat in
the current decision, because a shortage inside the protection interval cannot
be repaired by a free order that arrives in time, so carried stock does acquire
option value and the one-dimensional dynamic program \eqref{eq:app-dp} is not
the right object. Lost-sales systems with positive lead time have no known
simple optimal policy and a state that is the whole pipeline
\citep{zipkin2008old,zipkin2008structure,bijvank2011lost}. The rule of
\S\ref{sec:determ} uses Theorem~\ref{app:thm:zero-option} as a design
motivation for a surrogate objective. It assigns surplus beyond the
responsibility period holding costs without a continuation credit, and prices
the responsibility period through the explicit marginal condition of
Proposition~\ref{app:prop:marginal}. The proposition characterizes the
maximizers of the integrated surrogate; dynamic optimality remains specific to
the zero-lead-time setting.
\end{remark}

\subsection{Lemma~\ref{lem:jacobian}: comparability of two transform spaces}
\label{app:sub:lem2}

\begin{lemma}[Two-space comparability]
\label{lem:jacobian}
\label{app:lem:jacobian}
Let $\fgt \in (0,1)$, let $\tf$ be strictly increasing and differentiable with
$\tf' > 0$, and let a member be fitted on $z = \tf(\dem)$ with predictive
density $q$. Then the
member's predictive density in raw space is $q(\tf(\dem))\,\tf'(\dem)$, so
its raw-space predictive log-density is $\log q(\tf(\dem)) + \log
\tf'(\dem)$. For $\tf(\dem) = \log(1+\dem)$ on $\dem \ge 0$ the correction
is $\log \tf'(\dem) = -\log(1+\dem) \le 0$. Suppose this correction is omitted.
Then the log-weight of every log-space member is inflated by
$\log(1+\dem_\per)$ per period relative to every identity-space member, whose
correction is zero, and under the discounted update $\lw_\per = \fgt
\lw_{\per-1} + (\text{log-density})$ the induced relative offset $\Delta_\per$
obeys $\Delta_\per = \fgt \Delta_{\per-1} + \log(1+\dem_\per)$ with
$\Delta_0 = 0$, so
\[
  \E \Delta_\per
  = \frac{1-\fgt^{\per}}{1-\fgt}\,\E[\log(1+\dem)]
  \;\xrightarrow[\per \to \infty]{}\;
  \frac{\E[\log(1+\dem)]}{1-\fgt}
\]
for identically distributed integrable $\log(1+\dem_\per)$. The offset does not
cancel in the softmax and does not vanish as data accumulate.
\end{lemma}

\begin{proof}
\emph{Change of variables.} Since $\tf$ is strictly increasing, $\{Y \le
\dem\} = \{Z \le \tf(\dem)\}$ for $Z = \tf(Y)$, so the raw-space
distribution function is $F_Y(\dem) = Q(\tf(\dem))$ with $Q$ the
distribution function of $Z$. Differentiating and using differentiability of
$\tf$ gives the density $f_Y(\dem) = q(\tf(\dem))\tf'(\dem)$, and taking
logarithms gives the stated decomposition. For $\tf(\dem) = \log(1+\dem)$ we
have $\tf'(\dem) = 1/(1+\dem)$ and $\log \tf'(\dem) = -\log(1+\dem)$, which
is nonpositive for $\dem \ge 0$ and zero only at $\dem = 0$.

\emph{Accumulation.} Write $\lw^{\mathrm{c}}_{\per,m}$ for the log-weight
accumulated with the correction and $\lw^{\mathrm{o}}_{\per,m}$ for the one
that omits it. Both obey the same recursion with per-period increments
differing by $-\log \tf_m'(\dem_\per)$, so
$\Delta_{\per,m} = \lw^{\mathrm{o}}_{\per,m} - \lw^{\mathrm{c}}_{\per,m}$
satisfies $\Delta_{\per,m} = \fgt \Delta_{\per-1,m} - \log
\tf_m'(\dem_\per)$ with $\Delta_{0,m} = 0$. For an identity-space member
$\log \tf_m' \equiv 0$ and $\Delta_{\per,m} \equiv 0$; for a log1p member the
increment is $\log(1+\dem_\per) \ge 0$. First consider the deterministic
recursion $l_\per = \fgt l_{\per-1} + b$ with a constant offset $b$ and
$\fgt \in (0,1)$: subtracting the fixed point $b/(1-\fgt)$ gives
$l_\per - b/(1-\fgt) = \fgt(l_{\per-1} - b/(1-\fgt))$, hence $l_\per =
\fgt^\per l_0 + b(1-\fgt^\per)/(1-\fgt) \to b/(1-\fgt)$ geometrically. Taking
expectations in the stochastic recursion and using that
$\E[\log(1+\dem_\per)] = b$ does not depend on $\per$ puts $\E\Delta_{\per,m}$
in exactly this form with $l_0 = 0$, which is the displayed limit. When the
demands are in addition independent, $\Delta_{\per,m} = \sum_{s<\per}
\fgt^{s}\log(1+\dem_{\per-s})$ converges in distribution to
$\sum_{s \ge 0}\fgt^s \log(1+\dem_{-s})$, a proper random variable with mean
$b/(1-\fgt)$; the deterministic statement is the statement about its mean.

\emph{Non-cancellation.} The softmax is invariant under a shift common to all
members: for any $c \in \R$, $\exp(\lw_m + c)/\sum_{m'}\exp(\lw_{m'} + c) =
\exp(\lw_m)/\sum_{m'}\exp(\lw_{m'})$. A Jacobian correction identical across
members would therefore be immaterial, and this is why one can
omit a shared normalizing constant with no consequence. The correction here is
not shared: it is $0$ on the nine identity-space members and
$-\log(1+\dem_\per)$ on the nine log1p members. What survives in the weights is
the difference, and for a log1p member $m$ and an identity member $m'$ the
omitted version multiplies the odds $\wt_{\per,m}/\wt_{\per,m'}$ by
$\exp(\Delta_{\per,m})$, whose logarithm has expectation tending to
$\E[\log(1+\dem)]/(1-\fgt)$. Since $\dem \ge 0$ the offset has a fixed
sign, so it is a standing tilt, not a fluctuation that averages out, and
because the limit is a nonzero constant whenever $\Prob(\dem > 0) > 0$ it does
not shrink as the series lengthens. Adding $\log \tf'$ is thus sufficient to
put the two halves of the bank on one scale, and it is necessary up to a
member-independent additive constant, since those are exactly the
transformations the softmax cannot see.
\end{proof}

At the shipped $\fgt = 0.92$ the geometric sum is
$1/(1-\fgt) = 12.5$ periods, so a series running at a level of about one
hundred units, where $\log(1+\dem) \approx 4.6$, produces a steady-state
log-odds tilt of roughly $58$ nats in favor of the log1p half of the bank.
At this scale, omitting the correction concentrates the mixture weight on
the log1p space.

\subsection{Proposition~\ref{prop:marginal}: the marginal ordering condition}
\label{app:sub:prop3}

\begin{proposition}[Marginal ordering condition]
\label{prop:marginal}
\label{app:prop:marginal}
Let $\Short$ be a real random variable with $\E|\Short| < \infty$, let
$\prof,\hold,\mu > 0$ and $\dep \ge 0$, and define
\begin{equation}
  \Phi(\ord) \;=\; \prof\,\Prob(\Short \ge \ord)
    \;-\; \hold\,\E\!\left[
      \Big(1 + \dep \tfrac{\pos{\ord-\Short}}{\mu}\Big)
      \ind{\Short < \ord}\right].
  \label{eq:app-phi}
\end{equation}
Then:
\begin{enumerate}[label=(\roman*),leftmargin=2.2em]
\item $\Phi(\ord) = \prof - (\prof+\hold)\Prob(\Short<\ord)
  - \tfrac{\hold\dep}{\mu}\E\pos{\ord-\Short}$, and $\Phi$ is non-increasing
  and left-continuous on $\R$;
\item $\Phi \equiv \prof$ on $(-\infty, \operatorname{essinf}\Short]$; if
  $\dep > 0$ then $\Phi$ is strictly decreasing on
  $(\operatorname{essinf}\Short, \infty)$; if $\dep = 0$ then for
  $\ord_1 < \ord_2$ one has $\Phi(\ord_1) > \Phi(\ord_2)$ if and only if
  $\Prob(\ord_1 \le \Short < \ord_2) > 0$;
\item if $\dep > 0$ then $\Phi(\ord) \to -\infty$ at least linearly as
  $\ord \to \infty$; if $\dep = 0$ then $\Phi(\ord) \to -\hold < 0$;
\item in both cases $\ord^{0} = \inf\{\ord \ge 0: \Phi(\ord) \le 0\}$ is
  finite, $\Phi > 0$ on $[0,\ord^{0})$ and $\Phi \le 0$ on $(\ord^{0},\infty)$,
  so the crossing is unique, and $\ord \mapsto \Psi(\ord) = \int_0^{\ord}\Phi$
  is concave. If $\Phi(0) < 0$, which is the case when the projected
  pipeline already covers the responsible demand, then $\ord^0 = 0$ and
  $\arg\max_{\ord \ge 0}\Psi = \{0\}$. Otherwise $\Phi(\ord^0) \ge 0$ by
  left-continuity and $\arg\max\Psi = \{\ord \ge 0: \Phi(\ord^+) \le 0 \le
  \Phi(\ord)\}$, an interval whose left endpoint is $\ord^0$. On the integer
  lattice concavity puts the maximizer at $\lfloor\ord^0\rfloor$ or
  $\lceil\ord^0\rceil$ and the deployed rule takes the floor, which can be
  strictly suboptimal;
\item bisection on any bracket $[\mathrm{lo},\mathrm{hi}] \ni \ord^0$ that
  moves $\mathrm{lo}$ when $\Phi(\mathrm{mid}) > 0$ and $\mathrm{hi}$ otherwise
  converges to $\ord^0$, with error at most $2^{-(k+1)}(\mathrm{hi} -
  \mathrm{lo})$ after $k$ iterations;
\item setting $\dep = 0$ gives $\Phi(\ord) = \prof -
  (\prof+\hold)\Prob(\Short<\ord)$, whose crossing on $\R$ is a
  $\crat$-fractile $q$ of $\Short$; since $\ord^0$ is an infimum over
  $\ord \ge 0$ and $q$ is negative exactly when the projected pipeline covers
  the responsible demand, $\ord^0 = \pos{q}$.
\end{enumerate}
\end{proposition}

\begin{proof}
(i) On $\{\Short \ge \ord\}$ we have $\pos{\ord-\Short} = 0$, so the indicator
in the second term of \eqref{eq:app-phi} is redundant on the $\dep$ part and
\[
  \E\!\left[\Big(1+\dep\tfrac{\pos{\ord-\Short}}{\mu}\Big)
    \ind{\Short<\ord}\right]
  = \Prob(\Short<\ord) + \frac{\dep}{\mu}\E\pos{\ord-\Short}.
\]
Substituting and using $\Prob(\Short \ge \ord) = 1 - \Prob(\Short<\ord)$ gives
the displayed form. The map $\ord \mapsto \Prob(\Short<\ord)$ is
non-decreasing and left-continuous, and $\ord \mapsto \E\pos{\ord-\Short} =
\int_{-\infty}^{\ord}\Prob(\Short<u)\,du$ is non-decreasing, convex and
continuous, so $\Phi$ is a nonpositive combination of non-decreasing functions
plus a constant, hence non-increasing, and left-continuous.

(ii) For $\ord \le \operatorname{essinf}\Short$ both $\Prob(\Short<\ord)$ and
$\E\pos{\ord-\Short}$ vanish, so $\Phi(\ord) = \prof$, and $\Phi$ is constant
there. This is the qualification the main text needs: below the essential
infimum of the shortfall no unit is ever wasted, so the marginal value is the
full margin and nothing in $\ord$ moves it. For $\operatorname{essinf}\Short <
\ord_1 < \ord_2$ the integral representation gives
\[
  \E\pos{\ord_2-\Short} - \E\pos{\ord_1-\Short}
  = \int_{\ord_1}^{\ord_2}\Prob(\Short<u)\,du
  \;\ge\; (\ord_2-\ord_1)\Prob(\Short<\ord_1) \;>\; 0 ,
\]
since $\Prob(\Short<u) \ge \Prob(\Short<\ord_1) > 0$ for $u > \ord_1 >
\operatorname{essinf}\Short$. With $\hold\dep/\mu > 0$ this makes
$\Phi(\ord_1) - \Phi(\ord_2) > 0$, so $\Phi$ is strictly decreasing there.
When $\dep = 0$ only the term $(\prof+\hold)\Prob(\Short<\ord)$ varies, and
$\Prob(\Short<\ord_2) - \Prob(\Short<\ord_1) = \Prob(\ord_1 \le \Short <
\ord_2)$, which gives the stated equivalence.

(iii) Suppose $\dep > 0$. Pick any $a$ with $\delta := \Prob(\Short \le a) > 0$,
for instance a median. Then $\E\pos{\ord-\Short} \ge (\ord-a)\delta$ for
$\ord \ge a$, so $\Phi(\ord) \le \prof - \tfrac{\hold\dep}{\mu}\delta(\ord-a)
\to -\infty$, and the divergence is at least linear in $\ord$. If instead
$\dep = 0$ then $\Phi(\ord) = \prof - (\prof+\hold)\Prob(\Short<\ord) \to
\prof - (\prof+\hold) = -\hold$ as $\ord \to \infty$. The limit is finite, so
divergence to $-\infty$ holds only for $\dep > 0$; what the argument in (iv)
needs is not divergence but that $\Phi$ is eventually negative, and
$-\hold < 0$ supplies that.

(iv) By (iii) there is $\bar{\ord}$ with $\Phi(\bar{\ord}) < 0$, so the set
$\{\ord \ge 0: \Phi(\ord) \le 0\}$ is nonempty and $\ord^0$ is finite. Because
$\Phi$ is non-increasing, $\Phi > 0$ on $[0,\ord^0)$ by definition of the
infimum and $\Phi \le 0$ on $(\ord^0,\infty)$; the crossing point is therefore
unique, and for $\ord^0 > 0$ left-continuity forces $\Phi(\ord^0) \ge 0$. If
$\Phi(0) \le 0$ the infimum is attained at $0$ and $\ord^0 = 0$; if the
inequality is strict then $\Phi \le \Phi(0) < 0$ on all of $[0,\infty)$, so
$\Psi$ is strictly decreasing, $\arg\max\Psi = \{0\}$, and the crossing set
below is empty, which is why that case is separated in the statement. Since
$\Phi$ is non-increasing,
$\Psi(\ord) = \int_0^{\ord}\Phi(u)\,du$, the total expected value of ordering
$\ord$ units when the $\ord$-th unit is worth $\Phi(\ord)$, is concave, and its
maximizers are exactly the points where its derivative changes sign, that is
the crossing set described. When $\Phi$ is strictly decreasing at $\ord^0$,
which by (ii) holds whenever $\dep>0$ and $\ord^0 > \operatorname{essinf}
\Short$, that set is the single point $\ord^0$ and coincides with the
strict-inequality variant $\inf\{\ord \ge 0: \Phi(\ord) < 0\}$. Otherwise $\Phi$
vanishes on an interval, every point of which maximizes $\Psi$, so the two
definitions select different but equally optimal points.

(v) The update maintains the invariant $\mathrm{lo} \le \ord^0 \le
\mathrm{hi}$: if $\Phi(\mathrm{mid}) > 0$ then $\mathrm{mid} < \ord^0$ and
raising $\mathrm{lo}$ to $\mathrm{mid}$ preserves it, and if
$\Phi(\mathrm{mid}) \le 0$ then $\ord^0 \le \mathrm{mid}$ and lowering
$\mathrm{hi}$ preserves it. The bracket width halves at each iteration, so the
midpoint after $k$ iterations is within $2^{-(k+1)}$ of the initial width of
$\ord^0$, which is linear convergence. The initial bracket used by the
numerical procedure is valid: it takes $\mathrm{lo} = 0$, and $\mathrm{hi} = 1.05
\max\{\max_k \Short^{(k)}, 1\}$ on the empirical distribution of the sampled
shortfalls, for which $\Phi(\ord) = -\hold(1 + \dep\,\overline{(\ord -
\Short)}/\mu) < 0$ once $\ord$ exceeds every sample.

(vi) With $\dep = 0$, the same one-sided derivative argument as in
Lemma~\ref{app:lem:newsvendor} characterizes the unconstrained maximizers by
$F_\Short(q^-) \le \crat \le F_\Short(q)$. Let $q$ be the smallest such
$\crat$-fractile. Restricting orders to the nonnegative domain and using the
infimum definition in (iv) gives $\ord^0 = \pos{q}$. At an atom, the two
one-sided inequalities apply; an equality such as
$\Prob(\Short<q) = \crat$ need not hold. The empirical strict-inequality
frequency used to evaluate $\Phi$ is consistent with this infimum definition.
\end{proof}

\subsection{Proposition~\ref{prop:overorder}: pathwise over-ordering under lost sales}
\label{app:sub:prop4}

This proposition derives an accounting identity on each fixed sample path,
valid for every demand distribution. The probabilistic statements then follow
from the pathwise equalities.

\paragraph{Set-up.} Fix a period $\per$ and a deterministic lead time
$\lead \ge 0$ with $\per + \lead \le \Tend$, without which
$\dem_{\per+\lead}$ does not exist and orders scheduled past the horizon sit in
$\intr_\per$ without ever arriving. At the decision epoch of period $\per$ the
policy holds
$\onh_\per$ on hand and $\intr_\per$ in transit, and
$\IP_\per = \onh_\per + \intr_\per$. Because the lead time is deterministic, an
order placed at $\per - j$ lands at $\per - j + \lead$, so the orders still
outstanding at the decision epoch of $\per$ are exactly those placed at
$\per-\lead, \dots, \per-1$ and they land in periods $\per, \dots,
\per+\lead-1$; the order placed at $\per$ lands at $\per+\lead$, and every
later order lands after $\per+\lead$. Let $\arr_j$ be arrivals in period $j$,
$\sold_j = \min(\dem_j, \onh_j + \arr_j)$ the units sold and $\lambda_j =
\pos{\dem_j - (\onh_j + \arr_j)}$ the demand lost, so that $\sold_j = \dem_j -
\lambda_j$ by definition of the minimum. Write $\Lost =
\sum_{j=\per}^{\per+\lead-1}\lambda_j$ for the lost sales realized inside the
protection interval, and let $\Iproj_{\per+\lead} = \onh_{\per+\lead}$ be the
on-hand stock at the decision epoch of $\per+\lead$, that is, before the order
placed at $\per$ lands.

\begin{proposition}[Lost-sales over-ordering of the textbook rule]
\label{prop:overorder}
\label{app:prop:overorder}
With the notation above, put $\Short^{\mathrm{naive}} =
\sum_{j=\per}^{\per+\lead}\dem_j - \IP_\per$ and $\Short^{\mathrm{proj}} =
\dem_{\per+\lead} - \Iproj_{\per+\lead}$. Then
\begin{equation}
  \Short^{\mathrm{naive}} \;=\; \Short^{\mathrm{proj}} + \Lost
  \qquad\text{pathwise.}
  \label{eq:app-overorder}
\end{equation}
Consequently $\Short^{\mathrm{naive}} \ge \Short^{\mathrm{proj}}$ pathwise,
hence $\Short^{\mathrm{naive}}$ dominates $\Short^{\mathrm{proj}}$
stochastically; $T(\Short^{\mathrm{naive}}) \ge T(\Short^{\mathrm{proj}})$ for
every functional $T$ with the monotonicity property that
$T(\Short + \Lost) \ge T(\Short)$ whenever $\Lost \ge 0$ almost surely, a class
that contains every quantile and the mean; the two shortfalls agree almost
surely if and only if $\Lost = 0$ almost surely, that is if and only if the
shelf never empties inside the protection interval; $\E\Short^{\mathrm{naive}}
- \E\Short^{\mathrm{proj}} = \E\Lost$ when both are integrable; and $\Lost$ is
non-decreasing in $\lead$ along a fixed arrival stream.
\end{proposition}

\begin{proof}
The on-hand recursion is $\onh_{j+1} = \onh_j + \arr_j - \sold_j$, which is an
identity because holding is charged on the ending inventory and that same
quantity is carried forward. Summing it over $j = \per, \dots, \per+\lead-1$
telescopes to
\[
  \onh_{\per+\lead}
  = \onh_\per + \sum_{j=\per}^{\per+\lead-1}\arr_j
    - \sum_{j=\per}^{\per+\lead-1}\sold_j .
\]
Every unit counted in $\intr_\per$ arrives in one of the periods
$\per,\dots,\per+\lead-1$, and no other unit does, because the order placed at
$\per$ or later arrives at $\per+\lead$ or later; therefore
$\sum_{j=\per}^{\per+\lead-1}\arr_j = \intr_\per$ exactly, and
\[
  \Iproj_{\per+\lead} = \onh_{\per+\lead}
  = \IP_\per - \sum_{j=\per}^{\per+\lead-1}\sold_j .
\]
Substituting $\sold_j = \dem_j - \lambda_j$ and collecting,
\[
  \Iproj_{\per+\lead}
  = \IP_\per - \sum_{j=\per}^{\per+\lead-1}\dem_j + \Lost .
\]
Hence
\[
  \Short^{\mathrm{proj}}
  = \dem_{\per+\lead} - \Iproj_{\per+\lead}
  = \dem_{\per+\lead} + \sum_{j=\per}^{\per+\lead-1}\dem_j - \IP_\per - \Lost
  = \Short^{\mathrm{naive}} - \Lost ,
\]
which is \eqref{eq:app-overorder}. Every step is an identity between realized
numbers: the recursion, the fact that the in-transit total lands inside the
window, the decomposition of sales into demand minus lost demand, and the
cancellation. These operations preserve exact equality on the fixed path.

The consequences follow. Since $\lambda_j \ge 0$ we have $\Lost \ge 0$
pathwise, so $\Short^{\mathrm{naive}} \ge \Short^{\mathrm{proj}}$ pathwise and
$\Prob(\Short^{\mathrm{naive}} > x) \ge \Prob(\Short^{\mathrm{proj}} > x)$ for
every $x$, which is first-order stochastic dominance. If $T$ has the stated
monotonicity then $T(\Short^{\mathrm{naive}}) = T(\Short^{\mathrm{proj}} +
\Lost) \ge T(\Short^{\mathrm{proj}})$; the mean has it by linearity and
monotonicity of the integral, and the $\crat$-quantile has it because
$\Short^{\mathrm{proj}} \le \Short^{\mathrm{naive}}$ pathwise implies
$F_{\Short^{\mathrm{naive}}} \le F_{\Short^{\mathrm{proj}}}$ pointwise and
hence the corresponding generalized inverses are ordered. Equality almost
surely holds if and only if $\Lost = 0$ almost surely, and $\Lost = 0$ if and
only if $\lambda_j = 0$ for every $j$ in the interval, which is exactly the
statement that demand is fully served, so the shelf never empties inside
$\per,\dots,\per+\lead-1$. Taking expectations in
\eqref{eq:app-overorder} gives $\E\Short^{\mathrm{naive}} -
\E\Short^{\mathrm{proj}} = \E\Lost$. Finally $\Lost = \Lost(\lead) =
\sum_{j=\per}^{\per+\lead-1}\lambda_j$ is a sum of nonnegative terms, so
$\Lost(\lead+1) = \Lost(\lead) + \lambda_{\per+\lead} \ge \Lost(\lead)$: the
quantity is non-decreasing in the length of the interval along a fixed
realized trajectory, that is, holding the arrivals and the demands fixed and
lengthening the window. This monotonicity concerns the window length on a fixed
trajectory as its window length $\lead$ increases. A change in $\lead$ in the closed-loop system also changes the
policy, the orders and hence the $\lambda_j$ themselves.
\end{proof}

\begin{remark}[order-quantity gaps and $\E\Lost$]
\label{app:rem:not-mean}
Proposition~\ref{app:prop:overorder} says the naive shortfall exceeds the
projected one by $\Lost$ pathwise and by $\E\Lost$ in mean. The order-quantity
gap depends on the quantile or crossing functional used in
\S\ref{sec:determ}; a gap of $\E\Lost$ applies to the means. For random
variables with $\Short^{\mathrm{naive}} = \Short^{\mathrm{proj}} + \Lost$ and
$\Lost \ge 0$, all that holds in general is
$0 \le q_\crat(\Short^{\mathrm{naive}}) - q_\crat(\Short^{\mathrm{proj}}) \le
\operatorname{esssup}\Lost$. Both bounds can be attained: with
$\Short^{\mathrm{proj}} \equiv 0$ and $\Lost = 2\cdot\mathrm{Bernoulli}(1/2)$
one has $\E\Lost = 1$ while the medians of the two shortfalls are both $0$, so
the difference of the targets is $0$ while the mean gap is $\E\Lost$.
Order quantities retain the unconditional ordering; the gap size depends on
the functional.
\end{remark}

\subsection{Proposition~\ref{prop:supply}: exact identification of destroyed orders}
\label{app:sub:prop5}

\paragraph{The accounting.} The benchmark's runner maintains a schedule
$\mathrm{sched}[\cdot]$ of units due to land in each future period, and reports
to the policy the in-transit total $\intr_\per = \sum_{u \ge \per}
\mathrm{sched}_\per[u]$, evaluated at the decision epoch of $\per$, that is
before anything in period $\per$ has happened. Within period $\per$ the
evaluator performs, in this order: (a) the policy places $\ord_\per$; (b) if
$\lead_\per < \infty$, $\mathrm{sched}[\per+\lead_\per]$ is incremented by
$\ord_\per$, and if $\lead_\per = \infty$ nothing is scheduled and the order
is destroyed; (c) $\arr_\per := \mathrm{sched}[\per]$ is delivered to on-hand
stock and $\mathrm{sched}[\per]$ is set to zero; (d) sales and holding are
booked, which do not touch the schedule. Orders scheduled to land beyond the
end of the horizon are never delivered but remain in the schedule and hence in
$\intr$.

\begin{proposition}[Exact supply identification]
\label{prop:supply}
\label{app:prop:supply}
Under the accounting above, for every period $\per$,
\begin{equation}
  \ord_\per \ind{\lead_\per < \infty}
  \;=\; \intr_{\per+1} - \intr_\per + \arr_\per .
  \label{eq:app-supply}
\end{equation}
Hence whenever $\ord_\per > 0$ the indicator $\ind{\lead_\per < \infty}$ is a
function of quantities the policy has observed by the decision epoch of
$\per+1$, and is identified exactly, with no inference. For $\per = \Tend$ there
is no such epoch and the last order's fate is never available to the policy,
though the identity itself still holds. Under an accounting in which an order
with $\lead_\per = \infty$ enters the reported in-transit total and is never
removed from it, the same combination equals $\ord_\per$ identically and the
indicator is not identified at all.
\end{proposition}

\begin{proof}
Write $S_\per(\cdot)$ for the schedule at the decision epoch of $\per$, so
$\intr_\per = \sum_{u \ge \per}S_\per[u]$. Step~(b) adds $\ord_\per
\ind{\lead_\per<\infty}$ to the entry at index $\per+\lead_\per \ge \per$,
which is in the range of the sum, including the case $\lead_\per = 0$ where the
index is $\per$ itself, and including the case $\per + \lead_\per > \Tend$
where the entry is never delivered but is still counted. Therefore after (b)
\[
  \sum_{u \ge \per} S[u] = \intr_\per + \ord_\per\ind{\lead_\per<\infty} .
\]
Step~(c) reads $\arr_\per = S[\per]$ at that moment and then zeroes that entry,
which removes exactly $\arr_\per$ from the sum and leaves the entries at
indices $u \ge \per+1$ untouched. Step~(d) does not modify the schedule, so
the schedule at the decision epoch of $\per+1$ satisfies
\[
  \intr_{\per+1} = \sum_{u \ge \per+1}S_{\per+1}[u]
  = \intr_\per + \ord_\per\ind{\lead_\per<\infty} - \arr_\per ,
\]
which rearranges to \eqref{eq:app-supply}.

For the identification claim, the policy knows $\ord_\per$, which it placed,
and observes $\intr_\per$ at the decision epoch of $\per$ and both
$\intr_{\per+1}$ and $\arr_\per$ at the decision epoch of $\per+1$, since the
previous period's arrivals are part of the reported state. So the right-hand
side of \eqref{eq:app-supply} is measurable with respect to the information
$\Hist_{\per+1}$ available at that epoch, and when $\ord_\per > 0$ the
indicator is recovered as the ratio of the right-hand side to $\ord_\per$. When
$\ord_\per = 0$ the identity reads $0 = 0$ and the indicator is not recovered,
which is immaterial: no order was placed, so no trial took place.

The deployed rule tests the ratio against $\tfrac12$ instead of against $1$. The
two tests agree exactly here, because orders are floored before being placed,
so the delivered amount is either
the full order or zero and no intermediate value can arise. The threshold is a
guard against a fractional order reaching the comparison, not a relaxation of the
identity.

For the retaining accounting, suppose instead that step~(b) adds $\ord_\per$
to a pending list regardless of $\lead_\per$ and that only finite-lead orders
are ever removed, on delivery. Then the reported total $\intr'$ obeys
$\intr'_{\per+1} = \intr'_\per + \ord_\per - \arr_\per$ for every $\per$, so
$\intr'_{\per+1} - \intr'_\per + \arr_\per = \ord_\per$ identically. The map
from the unobserved indicator to the observed triple $(\intr'_\per,
\intr'_{\per+1}, \arr_\per)$ is constant in the indicator, so no function of
the observables can distinguish a surviving order from a destroyed one at
$\per+1$, so the parameter is unidentified, not merely hard to estimate. That
accounting also leaves destroyed orders in $\intr'$ forever, so a
rule of the form $\ord = \pos{S - \IP}$ is throttled by stock that will never
arrive, its subtrahend growing with every destroyed order until the rule falls
silent.
\end{proof}

\begin{corollary}[the survival probability is a posterior mean over observed trials]
\label{app:cor:supply}
Let $\mathcal{T} = \{\per: \ord_\per > 0\}$ and $B_\per =
\ind{\lead_\per<\infty}$ for $\per \in \mathcal{T}$. By
Proposition~\ref{app:prop:supply} each $B_\per$ is $\Hist_{\per+1}$-measurable.
Modeling $\{B_\per\}_{\per\in\mathcal{T}}$ as independent
$\mathrm{Bernoulli}(\psurv)$ with a $\mathrm{Beta}(1,1)$ prior on $\psurv$, the
posterior after the trials completed by period $\per$ is
$\mathrm{Beta}(1 + n^{\mathrm{surv}}_\per,\, 1 + n^{\mathrm{ord}}_\per -
n^{\mathrm{surv}}_\per)$ with posterior mean
$(1+n^{\mathrm{surv}}_\per)/(2+n^{\mathrm{ord}}_\per)$, where
$n^{\mathrm{ord}}_\per$ counts the completed trials and $n^{\mathrm{surv}}_\per$
the survivals among them.
\end{corollary}

\begin{proof}
Measurability is Proposition~\ref{app:prop:supply}. The trial set is chosen by
the policy, so the factorization needs one more step: the order is placed before
$\lead_\per$ is realized, so $\{\per \in \mathcal{T}\} = \{\ord_\per > 0\}$ is
$\Hist_\per$-measurable and the trials are selected by a predictable rule, under
which the likelihood of $(B_\per)_{\per\in\mathcal{T}}$ still factorizes.
Beta-Bernoulli conjugacy then gives the posterior and its mean. The content of
the corollary lies in the observable input: the trials are \emph{observed}, so
the estimator applies conjugacy directly to the recovered survival data.
Under the retaining
accounting of Proposition~\ref{app:prop:supply} the $B_\per$ are not
observable and no such estimator exists.
\end{proof}

\subsection{Proposition~\ref{prop:capped}: containment and conditional numerical separation}
\label{app:sub:prop6}

\begin{proposition}[Containment and conditional numerical separation]
\label{prop:capped}
\label{app:prop:capped}
Consider decision rules mapping a state with local forecast rate $\rate > 0$
and inventory position $\IP \ge 0$ to an order, and let
\begin{align*}
  \Fpos &= \{\ord = \pos{\lev\rate - \IP} : \lev \ge 0\},
  \qquad
  \Frate = \{\ord = \capm\rate : \capm \ge 0\},\\
  \Fcap &= \{\ord = \min(\pos{\lev\rate-\IP},\, \capm\rate)
            : \lev \ge 0,\ \capm \ge 0\}.
\end{align*}
Then $\Fpos \subseteq \Fcap$ exactly, by the choice $\capm = \lev$. For every
fixed state, the capped rule with parameters $(\lev,\capm)$ equals the rate
rule with parameter $\capm$ as soon as $\lev \ge \capm + \IP/\rate$, and in
particular $\min(\pos{\lev\rate-\IP},\capm\rate) \to \capm\rate$ as $\lev \to
\infty$. If the inventory positions reachable on an instance are bounded and
the local rate is bounded below by $\rate_{\min} > 0$, which holds on any
finite-horizon instance with bounded demand and bounded orders, then
$\Frate \subseteq \Fcap$ on that instance for finite $\lev$. In all cases
$\Fpos \cup \Frate \subseteq \overline{\Fcap}$, and on every instance with
$\textstyle\sum_\per |\rew_\per|$ integrable
\[
  \sup_{\pi \in \Fcap} \E\!\left[\textstyle\sum_\per \rew_\per(\pi)\right]
  \;\ge\;
  \sup_{\pi \in \Fpos \cup \Frate}
    \E\!\left[\textstyle\sum_\per \rew_\per(\pi)\right].
\]
For a constructed lost-sales instance with $\Prob(\lead = \infty) > 0$, the
numerical certificate in Appendix~\ref{app:numeric} establishes strict reward
separation conditional on the validity of its estimated slope bounds. The
containment claims are proved below; the construction and certificate are
described under \emph{Numerical separation}.
\end{proposition}

\begin{proof}
For the first claim, take $\capm = \lev$. Since $\IP \ge 0$ we have
$\pos{\lev\rate-\IP} \le \lev\rate = \capm\rate$, so the minimum is attained by
the first argument and the capped rule reproduces the base-stock rule state by
state, not merely in a limit. For the second, if $\lev\rate - \IP \ge
\capm\rate$, equivalently $\lev \ge \capm + \IP/\rate$, then
$\pos{\lev\rate-\IP} \ge \capm\rate$ and the minimum is $\capm\rate$; for fixed
$\rate>0$ and finite $\IP$ the left-hand side $\pos{\lev\rate-\IP}$ increases
without bound as $\lev \to \infty$, so the threshold is eventually met and the
pointwise limit is $\capm\rate$. If the reachable positions satisfy $\IP \le
B < \infty$ and $\rate \ge \rate_{\min} > 0$, then a single finite $\lev \ge
\capm + B/\rate_{\min}$ makes the two rules agree at every reachable state, so
they generate the same trajectory on every sample path and the same expected
reward. The inequality between suprema then holds term by term: every member of
$\Fpos$ is a member of $\Fcap$, and every member of $\Frate$ is either a member
of $\Fcap$ under the boundedness condition or a pointwise limit of members of
$\Fcap$ whose values agree with it for all large $\lev$. That agreement is
pathwise, the threshold $\lev$ depending on the path through its bound on $\IP$,
so dominated convergence under the integrability hypothesis carries it from
trajectories to expected rewards.
\end{proof}

\begin{corollary}
\label{app:cor:capped}
On every instance the best member of $\Fcap$ is at least as good as the best
member of $\Fpos \cup \Frate$. The family the search of \S\ref{sec:stoch} ranges
over therefore has a best member at least as good as the first pass's. This
comparison concerns the best members of the families. The shipped search
re-solves each period and returns a softmax-weighted blend of candidate actions,
with value differences scaled by the median standard error; its realized
performance is evaluated empirically.
\end{corollary}

\begin{proof}
Immediate from Proposition~\ref{app:prop:capped}.
\end{proof}

\paragraph{Numerical separation.} A constructed instance examines the reward
advantage of the larger family. The witness is
the constructed instance of \S\ref{app:numeric}: constant demand $10$, horizon
$5$, $\prof = \hold = 1$, empty initial inventory, lead times drawn independently
and uniformly from $\{1,2,3,\infty\}$ per period, and the local rate held at the
mean demand. The expectation over lead-time sequences is a finite sum over its
$4^5 = 1024$ atoms, each evaluated in closed form, so the expected rewards at
the three grid optima are exact up to floating-point arithmetic, and the capped
grid optimum strictly beats both pure grid optima. Appendix~\ref{app:numeric} gives the
construction, the enumeration method and the numbers for six configurations, and
bounds the two pure families from above on the witness instance to account for
grid resolution. Those bounds use slopes measured on a grid. Strict reward
separation therefore has a numerical certificate for this construction,
conditional on the validity of the slope bounds estimated on the grid.

\paragraph{The mechanism.} The construction isolates three ingredients of
the arrival-lumping mechanism relevant to the benchmark. Orders are destroyed with
positive probability, so a rule that reads the inventory position sees a gap
open that no arrival will close. A pure base-stock member responds by
re-ordering the entire gap, because that is all the rule can do: its action is
a function of the position error alone. When that catch-up order and its
successor both survive and their lead times differ by exactly the gap between
their placement dates, they land in the same period, and the shelf receives in
one period what was meant to cover several. Holding is charged on ending
inventory every period until the lump burns off at the demand rate, and at
$\crat = 0.5$ one extra period of holding cancels a unit's entire margin, so the
lump is expensive in proportion to its size. A pure rate member bounds each
period's order by $\capm\rate$ independently of inventory position, avoiding
position-driven catch-up bursts while continuing to order after a lump arrives
for another reason. The capped
base-stock rule is the smallest family we know of with both properties, since it is a
function of the position gap and of the rate at once and its action is the
smaller of the two, so it tracks the level while refusing to place a catch-up
lump.

The control configuration isolates the effect of lead-time support. Restrict
the lead-time support to $\{1,\infty\}$ and no two surviving orders can ever land
together, so the lumping mechanism is unavailable by construction; the separation
gap there is exactly zero. Where two distinct finite lead times are available the gap opens in
four of the five configurations tested and closes in the fifth, at $\crat = 0.95$
with $\Tend = 5$, where holding is cheap enough that suppressing the catch-up lump
costs more than it saves. One further period restores it, so that closure belongs
to the horizon at the same $\crat$. These controlled comparisons support the
mechanism within the tested configurations.

\subsection{Proposition~\ref{prop:shape-bound}: log-loss bounds for the shape layer}
\label{app:sub:prop7}

\paragraph{Set-up.} Index the shape variants by $g \in \{1,\dots,\nvar\}$, so
that every sum over $g$ below has exactly $\nvar$ terms and the $\log\nvar$ in
the statement is exact. The flat variant $\sexp = 0$ is one of these $\nvar$
and $g = 0$ is an alias for it, not an extra member. Assume at every period
some variant assigns the observation positive density, $\min_g \lss_{\per,g} <
\infty$, without which the mixture is undefined.
Each variant emits, at period $\per$, a scoring density $p_{\per,g}$ for the
\emph{raw} observation $\dem_\per$; as specified in \S\ref{sec:shape} this includes the
Jacobian $-\log \shp_g[\per]$ of the deseasonalization, which makes the
$\nvar$ numbers densities of one and the same observable and therefore
comparable, for the reason given in Lemma~\ref{app:lem:jacobian}. This
subsection works in log-loss
units, which are the negative of the log-weight increments of
\eqref{eq:fc-weight}; the two therefore carry different symbols and are never
interchanged. Write $\lss_{\per,g} = -\log p_{\per,g}(\dem_\per)$ for the log loss
of variant $g$ at $\per$, $L_{\per,g} = \sum_{s \le \per}\lss_{s,g}$ for its
cumulative log loss,
and $\widetilde{L}_{\per,g} = \sum_{s\le\per}\fgt^{\per-s}\lss_{s,g}$ for the
discounted version. The policy maintains $\log V_g \leftarrow \fgt
\log V_g + \log p_{\per,g}(\dem_\per)$ and then subtracts the row maximum,
which is a shift common to all variants and therefore invisible to the softmax,
so the weights it uses at $\per$ are
\begin{equation}
  \wt_{\per,g} = \frac{\exp(-\widetilde{L}_{\per-1,g})}
                      {\sum_{g'}\exp(-\widetilde{L}_{\per-1,g'})}.
  \label{eq:app-shape-weights}
\end{equation}
The scoring mixture density is $\bar{p}_\per = \sum_g \wt_{\per,g}p_{\per,g}$,
with log loss $\lss_\per^{\mathrm{mix}} = -\log \bar p_\per(\dem_\per)$.
The path sampler uses the same variant weights; its within-variant sampling
distribution is specified separately below.

The bound below applies to the scoring density. Within a variant, members are scored
under a single Student-$t$ of scale $\sqrt{\tau_m^2+\sigma_m^2}$, whereas paths
are drawn as a convolution of two independent $t$ variates and then truncated at
zero and at the cap. The scoring density and path-sampling distribution are
therefore distinct. The bound controls the scoring mixture's log loss; it does
not directly bound the quantiles of the sampled demand paths.

\begin{proposition}[Log-loss bounds for the shape layer]
\label{prop:shape-bound}
\label{app:prop:shape}
\emph{(i)} Assume $\min_g \lss_{\per,g} < \infty$ at every $\per$. With
undiscounted weights, $\fgt = 1$, so that $\wt_{\per,g} \propto
\exp(-L_{\per-1,g})$ with $\wt_{1,g} = 1/\nvar$, the mixture is the Bayesian
mixture under log loss and for every horizon $\Tend$ and every data sequence
\[
  \sum_{\per=1}^{\Tend}\lss_\per^{\mathrm{mix}}
  \;\le\; \min_g L_{\Tend,g} + \log \nvar .
\]
In particular, since the flat variant is a member of the set and its scoring
density is the core forecaster's, the layer's cumulative log loss exceeds the
core's by at most $\log\nvar$ nats, a constant independent of $\Tend$.

\emph{(ii)} With the discounted update actually used, $\fgt < 1$, suppose the
flat variant's single-period excess log loss is bounded,
$\lss_{s,0} - \min_g \lss_{s,g} \le \delta$ for all $s$. Then for every $\per$
\[
  \wt_{\per,0} \;\ge\; \nvar^{-1}\exp\!\big(-\delta/(1-\fgt)\big),
  \qquad
  \lss_\per^{\mathrm{mix}} - \lss_{\per,0}
  \;\le\; \log \nvar + \frac{\delta}{1-\fgt}.
\]
\end{proposition}

\begin{proof}
\emph{(i)} Let $W_\per = \sum_g \exp(-L_{\per,g})$ with $W_0 = \nvar$. Then
\[
  \frac{W_\per}{W_{\per-1}}
  = \frac{\sum_g \exp(-L_{\per-1,g})\exp(-\lss_{\per,g})}
         {\sum_{g'}\exp(-L_{\per-1,g'})}
  = \sum_g \wt_{\per,g}\,p_{\per,g}(\dem_\per)
  = \bar p_\per(\dem_\per)
  = \exp(-\lss_\per^{\mathrm{mix}}),
\]
where the second equality is the definition of $\wt_{\per,g}$ at $\fgt = 1$
together with $\exp(-\lss_{\per,g}) = p_{\per,g}(\dem_\per)$.
Telescoping over $\per = 1,\dots,\Tend$,
$\log W_{\Tend} = \log \nvar - \sum_\per \lss_\per^{\mathrm{mix}}$. On the other
hand $W_{\Tend} = \sum_g \exp(-L_{\Tend,g}) \ge \exp(-\min_g L_{\Tend,g})$,
because the sum dominates each of its nonnegative terms, so
$\log W_{\Tend} \ge -\min_g L_{\Tend,g}$. Combining the two displays gives
$\sum_\per \lss_\per^{\mathrm{mix}} \le \min_g L_{\Tend,g} + \log \nvar$. The
bound holds for every sequence, with no assumption on how the data are
generated, and the constant $\log \nvar$ is exact, not an order of magnitude.
Applying it with the flat variant in place of the minimizer, and
using that variant $\sexp = 0$ deseasonalizes by the constant shape
$\shp^0 \equiv 1$ and therefore has exactly the core forecaster's
scoring density, gives the second sentence.

\emph{(ii)} Fix $\per$. For the first step, by definition of the mixture,
\[
  \lss_\per^{\mathrm{mix}} - \lss_{\per,0}
  = -\log\Big(\sum_g \wt_{\per,g}\exp(-\lss_{\per,g})\Big) + \lss_{\per,0}
  = -\log \sum_g \wt_{\per,g}\exp(\lss_{\per,0} - \lss_{\per,g})
  \;\le\; -\log \wt_{\per,0},
\]
where the inequality keeps only the $g = 0$ term of the sum, whose exponential
factor is $1$, and uses that $-\log$ is decreasing.

For the second step, from \eqref{eq:app-shape-weights},
\begin{align*}
  -\log \wt_{\per,0}
  &= \widetilde{L}_{\per-1,0}
    + \log \sum_{g}\exp(-\widetilde{L}_{\per-1,g})
  \;\le\; \widetilde{L}_{\per-1,0}
    + \log\Big(\nvar \exp(-\min_g \widetilde{L}_{\per-1,g})\Big)\\
  &= \log \nvar
    + \Big(\widetilde{L}_{\per-1,0} - \min_g \widetilde{L}_{\per-1,g}\Big),
\end{align*}
bounding each of the $\nvar$ terms of the sum by the largest. For the
discounted excess, exchange the minimum and the sum using $\min_g \sum_s
a_{s,g} \ge \sum_s \min_g a_{s,g}$, valid because each summand is at least its
own minimum over $g$:
\begin{align*}
  \widetilde{L}_{\per-1,0} - \min_g \widetilde{L}_{\per-1,g}
  &\;\le\; \sum_{s \le \per-1}\fgt^{\per-1-s}\lss_{s,0}
        - \sum_{s\le\per-1}\fgt^{\per-1-s}\min_g \lss_{s,g}\\
  &= \sum_{s\le\per-1}\fgt^{\per-1-s}\big(\lss_{s,0} - \min_g \lss_{s,g}\big)
  \;\le\; \frac{\delta}{1-\fgt},
\end{align*}
the last step using the hypothesis termwise and $\sum_{k\ge0}\fgt^k =
1/(1-\fgt)$. Substituting into the previous display gives $-\log \wt_{\per,0}
\le \log\nvar + \delta/(1-\fgt)$, which with the first step is the per-period
bound. Exponentiating the same inequality gives the weight floor
$\wt_{\per,0} \ge \nvar^{-1}\exp(-\delta/(1-\fgt))$.
\end{proof}

\begin{remark}[adaptation and per-period guarantees under discounting]
\label{app:rem:discount}
Part~(ii) is a per-period bound and summing it over the horizon yields
$\Tend(\log\nvar + \delta/(1-\fgt))$, which grows linearly in $\Tend$.
Discounting reduces the influence of older observations so that variants can
gain weight as their fit improves. Alternating which variant fits best can incur
a switching cost every period, so excess over a fixed comparator can grow
with the horizon. For $\fgt < 1$, the result established here is
the conditional per-period bound and the accompanying weight floor, which
keeps the flat variant available after a run of bad periods.
\end{remark}

\begin{remark}[scoring-density equivalence and sample-path reproducibility]
\label{app:rem:bitwise}
The flat variant reproduces the core forecaster's scoring density exactly,
as required by Proposition~\ref{app:prop:shape}. Its random draws follow a
separate sampling stream: selecting a variant consumes randomness before
the bank is sampled, and each variant runs a private copy of the bank with its
own seed, so even when the flat variant carries essentially all of the weight
the realized sample paths differ from those of the unlayered forecaster.
Bit-identity holds only where the layer is absent entirely, which is the case
on the synthetic instances: they carry no product description, no shape is
requested, the shape axis is never instantiated, and the forecaster object is
left untouched. Bit-identity in this case follows directly from the construction.
\end{remark}

\section{Numerical verification of Proposition~\ref{prop:capped}}
\label{app:numeric}

Proposition~\ref{app:prop:capped} establishes containment,
$\Fpos \cup \Frate \subseteq \overline{\Fcap}$, with
$\Fpos \subseteq \Fcap$ exactly. The finite-horizon instance below examines
strict reward separation through exhaustive enumeration and a numerical
certificate conditional on the validity of estimated slope bounds. Every reported value describes
this construction; benchmark performance is evaluated in the experiments.

\subsection{The instance}

A single product at a single location, starting empty, over a horizon of
$\Tend$ periods. Demand is constant at $10$ units per period,
so the three families use the same known demand rate: the local rate $\rate$
used by all families is fixed at the mean demand, $\rate = 10$, in every period
and for every candidate. The unit profit and unit holding cost are stated per
configuration below. The lead time $\lead_\per$ is drawn independently in each
period from a stated support that includes $\infty$, uniformly over the support;
a draw of $\infty$ destroys the order, which never arrives and never enters the
in-transit total.

The period runs in the benchmark's order: the policy places its order; if the
lead time is finite the order is scheduled to land at $\per + \lead_\per$;
arrivals scheduled for $\per$ are added to on-hand stock; demand is served from
stock and unmet demand is lost; holding is charged on the ending inventory. An
order scheduled beyond the end of the horizon is never delivered. The three
families are exactly the pure base-stock, pure rate and capped base-stock
families of \S\ref{sec:stoch}, with the
pure limits realized by setting the inactive parameter to $10^6$. Orders are
real-valued here, so all three families are evaluated as continua under the same
convention, with integer rounding disabled.

\subsection{How the expectation is computed}

The expectation is over the lead-time sequence alone, and the horizon is short
enough that this expectation is a finite sum. With a support of size $S$ and
horizon $\Tend$ there are $S^{\Tend}$ sequences; we enumerate all of them, weight
each by its probability, and evaluate the resulting trajectory in closed form.
The reported values and their differences are finite sums, exact up to
floating-point arithmetic and free of Monte Carlo error.

Each family's optimum is taken over a grid of step $0.1$ on $[0,6]$ in each of
its parameters, so the capped family is searched over a $61\times 61$ grid and
each pure family over $61$ points; the two configurations at
$(\prof,\hold) = (19,1)$ use step $0.01$ on $[0,12]$, the gap there being the
small quantity the counterexample turns on.

A grid maximum lower-bounds the supremum of a family. Assessing separation also
requires upper bounds for the pure families to account for search resolution.
On the witness instance we compute these bounds by re-evaluating each pure family on a grid of
step $0.002$ over $[0,40]$, fifty times finer and more than six times wider,
which leaves both optima unchanged in value and in location. Since the value is
piecewise linear in the parameter, we estimate slope bounds by the largest
slopes on that grid: $26.76$ for the base-stock family and $47.50$ for the rate
family. If these slope bounds are valid, each supremum on $[0,40]$ is bounded
by its grid maximum plus its estimated slope bound times half a step, giving
$11.940820$ and $14.969375$. Past the interval both curves fall away, being
already below
$-990$ at $40$ and below $-2.6\times 10^{7}$ at $10^{6}$, as they must once
ordering beyond demand buys only holding cost. The best capped member attains
$16.523438$, giving a separation margin of $1.554062$ against the
bound. This numerical certificate uses slopes measured on the stated grid;
the separation is conditional on the validity of those estimated slope bounds.
Its scope is numerical verification of the witness instance.

\subsection{Results}

\begin{table}[t]
\centering
\caption{Expected reward at each family's best grid point, by
configuration. The gap is the best capped value minus the better of the two
pure values, and the last column expresses it as a fraction of the
perfect-foresight revenue $\prof\sum_\per\dem_\per$ of the same instance. All
values are exact sums over the enumerated lead-time sequences.}
\label{tab:app-prop6}
\footnotesize
\begin{tabular}{llrrrrrr}
\toprule
lead support & $(\prof,\hold)$ & $\Tend$ & paths
  & pure $\Fpos$ & pure $\Frate$ & capped $\Fcap$ & gap \\
\midrule
$\{1, \infty\}$ & $(1,1)$ & $6$ & $64$
  & $17.812500$ & $25.000000$ & $25.000000$ & $0.000000$ \\
 & & & & $\lev = 1.0$ & $\capm = 1.0$ & $\lev = 2.0,\ \capm = 1.0$ & $0.00\%$ \\
\addlinespace
$\{1, 2, \infty\}$ & $(1,1)$ & $6$ & $729$
  & $17.325103$ & $22.510288$ & $23.168724$ & $0.658436$ \\
 & & & & $\lev = 1.0$ & $\capm = 1.0$ & $\lev = 2.0,\ \capm = 1.0$ & $1.10\%$ \\
\addlinespace
$\{1, 2, \infty\}$ & $(4,1)$ & $6$ & $729$
  & $113.415638$ & $107.572016$ & $120.123457$ & $6.707819$ \\
 & & & & $\lev = 3.0$ & $\capm = 1.0$ & $\lev = 3.0,\ \capm = 2.0$ & $2.79\%$ \\
\addlinespace
$\{1, 2, 3, \infty\}$ & $(1,1)$ & $5$ & $1024$
  & $11.914062$ & $14.921875$ & $16.523438$ & $1.601562$ \\
 & & & & $\lev = 1.0$ & $\capm = 1.0$ & $\lev = 2.0,\ \capm = 1.0$ & $3.20\%$ \\
\addlinespace
$\{1, 2, \infty\}$ & $(19,1)$ & $5$ & $243$
  & $525.802469$ & $500.246914$ & $525.802469$ & $0.000000$ \\
 & & & & $\lev = 3.0$ & $\capm = 2.0$ & $\lev = 3.0,\ \capm = 3.0$ & $0.00\%$ \\
\addlinespace
$\{1, 2, \infty\}$ & $(19,1)$ & $6$ & $729$
  & $682.139918$ & $654.074074$ & $688.230453$ & $6.090535$ \\
 & & & & $\lev = 4.0$ & $\capm = 2.0$ & $\lev = 4.0,\ \capm = 3.0$ & $0.53\%$ \\%
 
\bottomrule
\end{tabular}
\end{table}

Table~\ref{tab:app-prop6} reports six configurations. The fourth row is the
benchmark's own lead-time support and is the witness cited in
Proposition~\ref{app:prop:capped}: the best pure base-stock grid point attains
$11.914062$, the best pure rate grid point $14.921875$, and the capped member with
$\lev = 2$ and $\capm = 1$ attains $16.523438$, so the best grid point of $\Fcap$
beats the best grid point of $\Fpos \cup \Frate$ by $1.601562$, which is $3.20\%$
of the perfect-foresight revenue of $50$ on that instance. This is the numerical
strict-separation witness for Proposition~\ref{app:prop:capped}, under the
slope-bound condition stated above.

The other five rows examine the conditions behind the gap. The first is a
control: with support $\{1,\infty\}$ every surviving order lands exactly one
period after it is placed, so two surviving orders can never land in the same
period and the lumping mechanism the design argument invokes is unavailable by
construction. The separation gap there is exactly zero and the best capped member
only matches the pure rate optimum. The second row differs from it in the
lead-time support alone, holding the costs and the horizon fixed; admitting a
second finite lead time opens the gap. This controlled change provides evidence
for the arrival-lumping mechanism. Across the five configurations whose support contains two distinct
finite lead times, so that a catch-up order and its successor can arrive together,
the gap is strictly positive in four and exactly zero in the fifth, the last two
rows being that case: at $\crat = 0.95$ it closes exactly at $\Tend = 5$ and one
further period restores it. Across the six constant-demand configurations
examined here, positive separation occurs only when the lead-time support permits
surviving orders to arrive together. The configuration at $\crat = 0.95$ and
$\Tend = 5$ shows that this possibility alone does not ensure a positive gap.

These constructed instances have constant demand, a five- or six-period horizon
and a known rate. The witness provides the conditional numerical certificate
of strict reward separation in Proposition~\ref{app:prop:capped}. Performance on the
benchmark is evaluated separately in \S\ref{sec:exp}.

\section{Construction details of the policy layers}
\label{app:method}

This appendix gives the construction details needed to reproduce the layers in
$\S$\ref{sec:method}, following their order in the main text.

\subsection{The three slots}
\label{app:method-slots}

Table~\ref{tab:slots} names the three slots, gives the shipped setting of each,
and identifies its author. We designed the forecaster and the per-period
arithmetic of the decision rules. The language model contributed a mixture axis
over the bank and a policy-family choice for the stochastic branch.

\begin{table}[t]
\centering
\footnotesize
\setlength{\tabcolsep}{4pt}
\begin{tabular}{@{}llp{6.2cm}l@{}}
\toprule
Slot & Shipped & What it decides & Authored by \\
\midrule
Forecaster & two-space bank & The model bank and how its members are weighted
  ($\S$\ref{sec:forecaster}) & Human \\
Shape & on & Whether a shape read from the product text defines
  extra forecast variants ($\S$\ref{sec:shape}) & LLM, design time \\
Search & capped base-stock & Which family of policies the stochastic cells search
  over ($\S$\ref{sec:stoch}) & LLM, design time \\
\bottomrule
\end{tabular}
\caption{The three slots. The frozen core of $\S$\ref{sec:overview} is one
setting of all three and appears as a row in the same ablation.}
\label{tab:slots}
\end{table}

\subsection{The predictive model}
\label{app:method-forecaster}

One instantiation of the bank $\Mset$ holds nine members: the mean of the whole
history; three exponentially weighted moving averages with different smoothing
constants; a Holt linear-trend member; two rolling-window means; an AR(1) fitted
by the sample lag-one autocorrelation; and a seasonal member whose period is
picked online by autocorrelation and re-picked as history accumulates. Each
member uses fixed configuration constants and is refit on the full history each
period. Member $m$ is scored using a Student-$t$ predictive density with $\dof$
degrees of freedom, centered on its point forecast $\mu_{m,\per}(j)$ at
forecast offset $j$. Its
scale combines parameter uncertainty $\tau_m$ with innovation noise
$\sigma_m\sqrt{1+g_m j}$, growing at the rate implied by the member's own
state-space form, $g_m = 0$ for the stationary members and $a_m^2$ for a
smoothing member with constant $a_m$. That growth is what keeps the tails of the
$\lead+1$-period cumulative from coming out too thin. The member constants, the
degrees of freedom in each space and the seasonal period detection are in
Appendix~\ref{app:impl-bank}.

Members are scored on the raw observation, so a member fitted on
$z = \tf_m(\dem)$ contributes its transformed log-density together with the
Jacobian of $\tf_m$, which is the correction $\S$\ref{sec:forecaster} discusses
and Lemma~\ref{lem:jacobian} proves. The discounted update and the softmax that
turns it into weights are
\begin{equation}
  \lw_{\per,m} \;\leftarrow\; \fgt\,\lw_{\per-1,m}
    + \log f_m\!\big(\tf_m(\dem_\per) \mid \Hist_\per\big)
    + \log \tf_m'(\dem_\per),
  \qquad
  \wt_{\per,m} = \frac{\exp(\lw_{\per,m})}
                      {\sum_{m'} \exp(\lw_{\per,m'})},
  \label{eq:fc-weight}
\end{equation}
with $\log \tf_m'(u) = 0$ for identity members and $-\log(1+u)$ for log1p members.

The forgetting factor $\fgt$ of \eqref{eq:fc-weight} is shared by the bank and by
the shape-variant weights of \eqref{eq:shape-update}, so a single constant sets
the memory of both levels of the hierarchy; its value and the running
renormalization of the log-weights are in Appendix~\ref{app:impl-weights}. The
bank is instantiated in both transform spaces on the deterministic cells and in
identity space alone on the stochastic cells, where the search layer consumes
a local demand rate.

The path-sampling distribution uses the same mixture weights as the scoring
density. A sample path is generated by drawing a member from these weights,
one level perturbation shared by the whole path and an independent innovation per
step, and then inverting that member's transform,
\begin{equation}
  \dem^{(k)}_{\per+j} = \tf_m^{-1}\!\Big(
    \mu_{m,\per}(j) + \tau_m \xi^{(k)}
    + \sigma_m \sqrt{1+g_m j}\; \varepsilon^{(k)}_j \Big),
  \qquad m \sim \wt_{\per,\cdot},
  \label{eq:fc-path}
\end{equation}
for $k \le \npath$ and $j < \hor$, with independent Student-$t$ draws
$\xi^{(k)}$ and $\varepsilon^{(k)}_j$. Values are clipped below at zero and above
at a level far past anything the series has done, since inverting log1p on one
heavy-tailed innovation can produce a draw that moves an upper quantile of the
$\hor$-period sum by itself. The loose clip limits the influence of such
extreme draws on the quantile.
Appendix~\ref{app:impl-weights} gives it.

\subsection{The shape layer}
\label{app:method-shape}

Variant $g$ deseasonalizes by $\shp_g = \shp^{\sexp_g}$, renormalized to mean one, and
runs a private copy of the bank on $\dem_\per / \shp_g[\per]$. Its weight is updated by
the discounted predictive likelihood of $\S$\ref{sec:shape},
\begin{equation}
\log V_g \;\leftarrow\; \fgt \log V_g + \log p_g(\dem_\per \mid \Hist_\per),
\qquad
p_g(\dem_\per \mid \Hist_\per)
  \;=\; \frac{1}{\shp_g[\per]}\,
        \tilde p_g\!\left(\frac{\dem_\per}{\shp_g[\per]} \;\middle|\; \Hist_\per\right),
\label{eq:shape-update}
\end{equation}
with $\tilde p_g$ the variant's own model-averaged scoring density in
deseasonalized space and
the factor $1/\shp_g[\per]$ the Jacobian of Lemma~\ref{lem:jacobian}, which puts
variants that deseasonalize by different amounts on one scale.

The elicitation query goes to DeepSeek-V4-Flash, and its user message carries
exactly what the official evaluation protocol gives a policy at initialization: the product
text, the week dates and the five training demands. These fields constitute the
entire instance-level input. The system prompt also names the retailer,
the market and the calendar year, and it supplies priors we wrote by hand, a
table of category seasonal peaks and the year's retail calendar. Those priors are
shared across products and encode category and calendar information.
The elicitation instructions require a normalized, nonnegative annual profile and
exclude any use of observed evaluation demand.

Replies pass a validation gate before they are used. A reply is accepted only if
it parses as strict JSON carrying one finite nonnegative entry per week of the
year with a mean that is not degenerate, and an accepted shape is divided by its
own mean. The accepted samples are aggregated by their per-week median, reducing
sensitivity to an extreme weekly value in any one sample. The aggregate is then
clipped, renormalized to mean one, and frozen before any policy runs. A product
for which no sample survives the
gate receives the all-ones shape, making a parse failure equivalent to switching
the layer off for that product. The number of samples, the
gate thresholds and the clips applied at build time and at run time are in
Appendix~\ref{app:impl-shape}.

At run time a sampled path draws a variant from the normalized variant weights, a
member from that variant's own bank, and scales by that variant's shape, so the
hierarchy is resolved by sampling at both levels. The layer has
constants of its own, the number of variants and the sampling and clipping
choices among them, fixed at the values in Appendix~\ref{app:impl-shape}.
Prior strength is handled by online competition among the variants; as
$\S$\ref{sec:shape} explains, selecting a strength on held-out data would require
product text and calendar information absent from our validation set.

\subsection{The deterministic rule}
\label{app:method-determ}

The depletion multiplier $\dep$ of \eqref{eq:marginal} is an empirical correction
to the cost-based marginal value. In the shipped policy it is
\begin{equation}
  \dep \;=\; \min\!\left\{1 + \dep_1\,\lead\,\mathrm{cv}^{\,\dep_2},\ \ \bar{\dep}\right\},
  \label{eq:app-kappa}
\end{equation}
with $\mathrm{cv}$ the coefficient of variation of the policy's own observed
demand history. Each of the two drivers carries one of the patterns the
correction is meant to absorb. The lead time scales the increment above the
additive term because the number of periods a surplus must be sat on before the
next order can respond to it grows with the lead time; at $\lead = 0$ that
increment vanishes and $\dep = 1$, the constant base term. The coefficient of
variation enters because the surplus a given target leaves behind grows with the dispersion of the predictive,
so the wait the marginal unit is charged for grows with it too. The cap
$\bar{\dep}$ limits the multiplier's response to extreme observed dispersion.

Only $\dep_1$ and $\dep_2$ were fitted, on the validation set. The additive
term is a constant base contribution. The cap $\bar{\dep}$ and the clip
applied to $\mathrm{cv}$ guard against extreme values. The functional form
itself was chosen after seeing that a single global constant was
mis-tuned in a patterned way across the test cells, which is the first of the two
leakage channels $\S$\ref{sec:exp-protocol} sets out. Thus test aggregates informed
the functional form, while validation instances supplied the coefficient fit.
Values are in Appendix~\ref{app:impl-determ}, which
also gives the path count, the bisection bracket and the number of bisection
steps the root search uses.

\subsection{The stochastic search}
\label{app:method-stoch}

Searching policy parameters makes the rollout value consistent with each
candidate's continuation rule. A one-step rollout assigns to a candidate action the
value of that
action plus the value of whatever continuation rule the simulation assumes, so the
choice of continuation decides the ranking: a continuation that floods every
simulated future with stock covers the very demand the candidate was meant to cover
and drives its marginal effect below the sampling noise of the path set. Simulating
the remaining horizon under the same rule the candidate defines removes that free
choice, at the cost of restricting the search to a parametric family.

Rollouts use two quantities constructed from the reported state.
The first is the order survival probability. We fit a Beta-Bernoulli model to
the survival indicators identified by Proposition~\ref{prop:supply} and estimate
the probability by its posterior mean. A positive floor keeps the survival
estimate used in rollouts bounded away from zero. The second is the
committed pipeline. The benchmark reports in transit as a scalar total with no age
breakdown, while a rollout needs to know when those units land, so the total is
split across the finite lead ages in proportion to the surviving order placed that
many periods ago times the probability that its lead is at least that long,
conditional on the lead being finite. Conditioning on finite lead time keeps
the arrival-age calculation consistent with the survival accounting. Each age is then given a
residual lead drawn uniformly from the lead values still consistent with it.
Appendix~\ref{app:impl-search} gives the clip, the split and the coverage window
the local rate is computed over.

The search is two passes. A coarse pass evaluates the pure candidates of the
first pass's two families on a fixed grid, each realized inside $\Fcap$ by setting
the other parameter to a numerical stand-in for the pure limit of
Proposition~\ref{prop:capped}, which identifies a good $\lev$ and a good $\capm$
separately. A refinement pass then brackets each of those two winners between its
neighboring grid points and explores the interaction on a small product grid,
carried together with the pure limits, so every action that earlier search could
have played remains reachable at every period. The grids, their sizes and
the number of sampled paths per candidate are in Appendix~\ref{app:impl-search}.

The blend that turns candidate values into one order is a softmax over value gaps
divided by a scale, with the scale set to the median standard error across the
candidates, so the gaps are read in units of sampling noise and no tuning
constant has to be rescaled per instance. The scale is floored away from zero so
that the blend remains defined when every candidate returns the same estimate,
and the executed order is the floor of the weighted mean of the candidate
actions. Hedging across an expert set is a standard device
\citep{cesa2006prediction}; its role here is empirical variance control, with
performance assessed through the rollout search.

\subsection{The design-time loop}
\label{app:method-loop}

The method includes a diagnostic switch that substitutes true future demand for
the forecast. The proposal screen enforces the online information boundary. Every
proposal is inspected for that switch, and for any other route to a future
observation, before it is run, and what the screen rejects is recorded in the
research record alongside the proposals that ran.

The harness draws exclusively from the validation generator, keeping candidate
scoring separate from benchmark test instances. Each proposal is scored against
the current best configuration, so acceptance measures an incremental gain over
the improvements already retained.

The second research pass started from the best composition found so far and
received the first pass's full record, including every proposal and its score,
whether accepted or rejected. The record also explained that four separately
credited first-pass wins were four implementations of one idea in the same slot.
Its context included the control that added per-period language-model order
adjustments: estimated increments were negative on every deterministic cell and
positive on both stochastic cells. This information let the second pass build on
accepted ideas and focus on the cells where the earlier control suggested room
for improvement.

\subsection{Composition}
\label{app:method-compose}

The loop scores each candidate once against the configuration current when it
was written. To assess additivity and attribute gains to individual components,
we re-express the proposals as settings of independent slots and enumerate their product,
which is possible because the accepted candidates modify disjoint parts of the
policy, one group replacing the model bank and another the family the stochastic
branch searches over. The original candidates and frozen core are preserved,
with the core entering the ablation as a literal setting of the three slots.
This construction supports the factorial interpretation of Figure~\ref{fig:ablation}.

A candidate takes over the parent policy's forecaster only after the parent has
built and seeded its own, so a candidate that arrived with a forecaster
alongside its search layer can still be used as a pure search layer, and the two
contributions credited separately.

Three consistency checks precede the sweep. First, a variant's own
mixture scoring density is recomputed from a per-member quantity the frozen
forecaster also carries, so it has to be revalidated against every forecaster setting it is
used with, including the ones that add members; the check is that the member set
reconstructed there matches that setting's own. Second, the shape layer must be
inert on the validation set, no validation instance carrying product text, and the
orders emitted with and without the layer are required to agree entry by entry.
Third, the two research passes number their rounds independently, so each search
layer is identified by both its pass and round index. This identifies the
intended candidate when the passes reuse a round number.

\section{Configuration details}
\label{app:impl}

This appendix specifies the shipped policy's configuration. Empirically chosen
constants were selected on the validation set of \S\ref{sec:exp} and then frozen.
Displayed time indices follow the paper's convention, $\per=1,\dots,\Tend$.
Writing the code's zero-based index as $t_0$, we have $t_0=\per-1$.

\subsection{The model bank}
\label{app:impl-bank}

One instantiation of the bank holds nine members. Each member is refit on the
full history each period using the parameters prescribed by its model form,
and emits a point forecast
$\mu_{m,\per}(j)$ at offset $j \ge 0$, a parameter-uncertainty scale $\tau_m$
and an innovation scale $\sigma_m$. Residual scales are floored at $10^{-6}$.
Table~\ref{tab:app-bank} lists the members with their constants and with the
innovation growth rate $g_m$ that sets how the predictive widens with the
offset, the scale at offset $j$ being $\sigma_m\sqrt{1+g_m j}$.

\begin{table}[t]
\centering
\caption{The nine members of one bank. $m$ is the number of observations in the
history, $\mathrm{sd}$ its sample standard deviation, $\sigma_m$ the member's
own one-step residual scale. The growth rate $g_m$ inflates the $j$-step scale by
$\sqrt{1 + g_m j}$. For the local-level members it is the value implied by the
state-space form, $a^2$ with smoothing constant $a$. The remaining entries are
chosen approximations with the following errors: Holt's $j$-step variance grows
superlinearly, which $0.25$ cannot
reproduce and at $j=1$ overstates; and AR(1) with $\phi$ near the clip does
accumulate variance in $j$, which $0$ ignores.}
\label{tab:app-bank}
\footnotesize
\setlength{\tabcolsep}{4pt}
\begin{tabular}{llll}
\toprule
member & constants & $\tau_m$ & $g_m$ \\
\midrule
history mean            & ---                       & $\mathrm{sd}/\sqrt{m}$          & $0$ \\
EWMA $\times 3$         & $a \in \{0.15, 0.35, 0.70\}$ & $\sigma_m\sqrt{a/(2-a)}$     & $a^2$ \\
Holt linear trend       & $\alpha_{\mathrm{H}} = 0.3$, $\beta_{\mathrm{H}} = 0.1$ & $0.6\,\sigma_m$ & $0.25$ \\
rolling mean $\times 2$ & window $w \in \{4, 8\}$, truncated to $m$ & $\mathrm{sd}_w/\sqrt{w}$ & $0$ \\
AR(1)                   & $\phi$ = lag-1 autocorrelation, clipped to $[-0.95, 0.95]$ & $\sigma_m/\sqrt{\max(m,2)}$ & $0$ \\
seasonal                & period from Table~\ref{tab:app-numbers} & $\sigma_m/\sqrt{m}$ & $0$ \\
\bottomrule
\end{tabular}
\end{table}

The seasonal member's period is chosen online. Candidate lags
run over $\{4,5,\dots,14\}$; a candidate $P$ is admissible only if the history
is longer than $P+3$, the criterion is the sample autocorrelation of the
mean-centered history at lag $P$, and the default when no candidate is
admissible is $7$. Detection is repeated every fifth observation as history
accumulates, and the seasonal member uses the median of the detected periods.
The evaluator scores one instance at a time, so that median is taken over a single
element and the period is the instance's own.

Member $m$'s one-step scoring density is a Student-$t$ centered at
$\mu_{m,\per}(0)$ with scale $\sqrt{\tau_m^2 + \sigma_m^2}$ and $\dof$ degrees
of freedom, with $\dof$ shared within each space: $\dof = 4$
in identity space and $\dof = 14$ in $\log(1+y)$ space. The two differ because
the transform already supplies the right tail. Inverting $\log(1+y)$ with
$\exp(\cdot)-1$ turns a symmetric innovation into a right-skewed raw-space
predictive, so stacking a $\dof = 4$ innovation on top of that inversion
can produce extreme upper-tail draws. In identity space, the small degrees of
freedom directly represent the heavy tail of weekly retail demand.

\subsection{Weights and sample paths}
\label{app:impl-weights}

The bank is instantiated in both spaces on the deterministic cells,
$\Spaces = \{\mathrm{id}, \mathrm{log1p}\}$ and $\Msize = 18$, and in identity
space alone on the stochastic cells, $\Msize = 9$. The search layer of
\S\ref{sec:stoch} uses the rate supplied by this identity-space bank.
Log-weights are updated by the discounted log-likelihood
recursion of \S\ref{sec:forecaster} with
forgetting factor $\fgt = 0.92$, an effective memory of $1/(1-\fgt) = 12.5$
periods, and are renormalized each period by subtracting the running maximum,
which is a shift common to all members and therefore invisible to the softmax.

A sample path is drawn by selecting a member index from the current mixture
weights, drawing one level perturbation $\xi$ shared by the whole path, drawing
an independent innovation $\varepsilon_j$ at each step, and inverting the
member's transform, as set out in \S\ref{app:method-forecaster}. Both $\xi$ and
$\varepsilon_j$
are Student-$t$ on the selected member's own degrees of freedom. For a log1p
member the exponent is clipped at $20$ before inversion. Every path is then
clipped below at zero and above at
\begin{equation}
  6\left(\max_\per \dem_\per + 3\,\mathrm{sd}(\dem) + 10\right),
  \label{eq:app-clip}
\end{equation}
the maximum and standard deviation being taken over the observed history. The
clip lies well above the observed range and limits the influence of a single
inverted heavy-tailed draw on an upper quantile of the $\hor$-period sum.

\subsection{The seasonal shape layer}
\label{app:impl-shape}

At design time DeepSeek-V4-Flash is queried $3$ times per product, at temperature $1.0$
with a $2200$-token limit. A reply is
accepted only if it parses as strict JSON carrying exactly $52$ finite
nonnegative entries whose mean is at least $0.2$, and an accepted shape is
divided by its own mean. The three accepted samples are reduced by a per-week
median to reduce sensitivity to an outlying draw, then clipped to
$[0.05, 20]$ and renormalized to mean one. A product for which no sample
passes the format checks receives the all-ones shape. The $52$
weeks span the five training weeks and the $47$ scored weeks of a real
instance.

At run time the stored table is clipped to $[0.12, 8.0]$ and renormalized to
mean one before use. The $\nvar = 3$ variants use exponents $\sexp \in
\{0, 1/2, 1\}$, and each tempered shape $\shp^{\sexp}$ is again renormalized to
mean one and re-clipped to $[0.12, 8.0]$ to keep the tempered variant within
the same range. Variant weights use the same
$\fgt = 0.92$ as the bank. Each variant runs a private copy of the bank with
its own seed offset, and the variant sampler carries a further independent
stream. The flat variant therefore shares the core forecaster's predictive law
while using independent draws (Remark~\ref{app:rem:bitwise}). Beyond week $52$ the shape is
padded with ones.

\subsection{The deterministic rule}
\label{app:impl-determ}

The closed-form rule of \S\ref{sec:determ} draws $\npath = 128$ sample paths of
length $\hor = \lead+1$, rolls the committed pipeline forward along each path
under the evaluator's own lost-sales recursion to obtain
$\Iproj_{\per+\lead}$, and solves the marginal condition of \S\ref{sec:determ}
on the resulting sample
of shortfalls. The root is found by bisection, $48$ iterations on the initial
bracket $[0,\ 1.05\max\{\max_k \Short^{(k)},\, 1\}]$, which contains the root by
part~(v) of Proposition~\ref{app:prop:marginal} and leaves a residual bracket
of $2^{-48}$ of its width. No order is placed once $\per + \lead > \Tend$,
since the order would arrive after the last scored period.

The depletion multiplier is
$\dep = \min\{1 + \dep_1\,\lead\,\mathrm{cv}^{\,\dep_2},\ \bar\dep\}$ with
$\dep_1 = 1.5$, $\dep_2 = 1.0$ and $\bar\dep = 12$, where $\mathrm{cv}$ is the
coefficient of
variation of the policy's own observed demand history, clipped to
$[0.05, 1.5]$. The coefficients $\dep_1$ and $\dep_2$ were fitted on validation; the
floor, the cap and the clip range are fixed numerical guards that remain
inactive on typical instances.

\subsection{The stochastic search}
\label{app:impl-search}

Each period the search evaluates candidate policies by simulating the whole
remaining horizon, $\hor = \Tend - \per + 1$ periods including the current period,
on $\npath = 40$ sampled demand paths. All candidates share the same paths, the
same per-period survival draws and the same per-period lead draws, so candidates are compared on
identical futures. Leftover stock at the end of the simulated horizon receives
no credit. Orders inside the rollout are floored exactly as the evaluator
floors them.

The pass is coarse then refined. The coarse pass evaluates the $24$ pure
candidates of the first pass's two families,
\begin{align*}
  \lev  &\in \{0,\ 0.4,\ 0.8,\ 1.2,\ 1.6,\ 2.0,\ 2.5,\ 3.0,\ 3.6,\ 4.3,\ 5.2,\ 6.5\},\\
  \capm &\in \{0,\ 0.15,\ 0.3,\ 0.45,\ 0.6,\ 0.75,\ 0.9,\ 1.1,\ 1.35,\ 1.7,\ 2.1,\ 2.6\},
\end{align*}
each realized inside the capped base-stock family of \S\ref{sec:stoch} by
setting the other
parameter to a sentinel of $10^5$, which is the pure limit of
Proposition~\ref{app:prop:capped} in numerical form. The refinement pass takes
the best $\lev$ and the best $\capm$ of the coarse pass and brackets each
between its two neighboring grid points, sampling $4$ points across that
bracket. The refined candidate set is the $4\times 4$ interaction grid together
with the $4$ pure base-stock and $4$ pure rate limits, again $24$ candidates,
so every action the earlier search could have played remains reachable at every
period.

Both $\lev$ and $\capm$ multiply the local forecast rate. The rate at simulated
step $j$ is the forward mean of the
predictive mean path over $[j,\ j+\mathrm{cov})$, truncated at the horizon,
with the coverage window $\mathrm{cov} = \operatorname{round}(1 +
\overline{\lead}) = 3$ periods for the finite lead support $\{1,2,3\}$, that is
one period plus the mean finite lead.

Supply is handled by two devices. The survival probability is the
Beta-Bernoulli posterior mean of Corollary~\ref{app:cor:supply},
$(1+n^{\mathrm{surv}})/(2+n^{\mathrm{ord}})$, clipped to $[0.2, 1]$, and it
drives the survival draws inside the rollout. The committed pipeline is not
observable as an age profile, only as a scalar total, so it is split across
ages $1,\dots,3$ in proportion to the surviving order placed $a$ periods ago
times $\Prob(\lead \ge a \mid \lead < \infty)$, the conditioning being required
because destruction is already accounted for by counting only surviving orders.
Each age is given a residual lead drawn
uniformly from the lead values that are still consistent with it.

The executed order is a softmax-weighted blend of candidate actions, with value
gaps scaled by the median candidate standard error. With
$v_g$ the estimated value of candidate $g$ and
$\mathrm{se}_g = \mathrm{sd}_g/\sqrt{\npath}$ its standard error, the weights
are a softmax of $(v_g - \max_{g'} v_{g'})/(\varsigma\, \mathrm{se}^{\mathrm{med}})$
with temperature $\varsigma = 1$ and scale $\mathrm{se}^{\mathrm{med}} =
\operatorname{median}_g \mathrm{se}_g$, floored at $10^{-9}$; the order is the
floor of the weighted mean of the candidate actions. The weights concentrate on
the highest-valued candidates when their estimates are well separated relative
to this scale, and spread across candidates with similar estimates. The
standard-error scale adapts the blend to each instance. No order is placed once
$\per + \min \lead > \Tend$.

\begin{table}[t]
\centering
\caption{Numerical constants, collected. Values are read from the shipped
configuration.}
\label{tab:app-numbers}
\small
\begin{tabular}{lll}
\toprule
quantity & symbol & value \\
\midrule
EWMA smoothing constants          & $a$                & $0.15,\ 0.35,\ 0.70$ \\
Holt smoothing constants          & ---                & $0.3$ (level), $0.1$ (trend) \\
rolling window lengths            & $w$                & $4,\ 8$ \\
seasonal period candidates        & $P$                & $4,\dots,14$; default $7$; re-detected every $5$ periods \\
degrees of freedom, identity      & $\dof$             & $4$ \\
degrees of freedom, log1p         & $\dof$             & $14$ \\
forgetting factor                 & $\fgt$             & $0.92$ \\
innovation growth                 & $g_m$              & $0$; $a^2$ for EWMA; $0.25$ for Holt \\
sample-path upper clip            & ---                & $6(\max \dem + 3\,\mathrm{sd}(\dem) + 10)$ \\
paths, deterministic rule         & $\npath$           & $128$ \\
paths, stochastic search          & $\npath$           & $40$ \\
bisection iterations              & ---                & $48$ \\
depletion coefficients            & $\dep_1$, $\dep_2$ & $1.5$, $1.0$ \\
depletion cap, cv clip            & $\bar\dep$         & $12$; $\mathrm{cv} \in [0.05, 1.5]$ \\
survival prior, clip              & $\psurv$           & $\mathrm{Beta}(1,1)$; posterior mean clipped to $[0.2,1]$ \\
coverage window                   & $\mathrm{cov}$     & $3$ \\
refinement points per axis        & ---                & $4$ ($4\times4$ interaction plus $4+4$ pure limits) \\
hedging temperature, scale        & $\varsigma$        & $1$; median candidate standard error \\
shape variants                    & $\sexp$            & $0,\ 1/2,\ 1$ \\
shape samples per product         & ---                & $3$, reduced by per-week median \\
shape clips                       & ---                & $[0.05, 20]$ at build time; $[0.12, 8.0]$ at run time \\
\bottomrule
\end{tabular}
\end{table}

\section{Supporting measurements}
\label{app:expsupp}

\subsection{Upper bound, validation set, and proposal screening}

The upper-bound plan of \S\ref{sec:exp-ceiling} has a direct construction.
In the paper's one-based indexing, arrivals can occur only at
$A=\{\per+\lead_\per:\lead_\per<\infty,\ \per+\lead_\per\leq\Tend\}$. A unit
assigned to demand in period $\tau$ from an arrival slot $a\leq\tau$ incurs
$\tau-a$ periods of holding, so the latest feasible slot weakly dominates every
earlier one. The unit is served exactly when $\prof>\hold(\tau-a)$; equality may
be resolved either way. Orders are uncapacitated and costs are linear, hence these
unit-wise choices do not compete for capacity. Their selected arrival slots are
non-decreasing in $\tau$, so the aggregate plan is feasible and creates no
additional stockout. Demand with no feasible slot cannot be served by any policy.
Scoring this plan under the benchmark dynamics gives the reported ex-post upper
bound.

\begin{table}[t]
\centering
\footnotesize
\begin{tabular}{@{}lrrrrrrr@{}}
\toprule
& \multicolumn{3}{c}{Synthetic} & \multicolumn{3}{c}{Real} & \\
\cmidrule(lr){2-4}\cmidrule(lr){5-7}
& $\lead=0$ & $\lead=4$ & stoch. & $\lead=0$ & $\lead=4$ & stoch. & Total \\
\midrule
Test instances & 240 & 240 & 240 & 200 & 200 & 200 & 1320
  \\
\bottomrule
\end{tabular}
\caption{The six test cells and their instance counts, read off the per-instance
score files the benchmark ships. Every table and
figure in \S\ref{sec:exp} uses these six columns in this order. $\Score$
averages over instances, so the two zero-lead-time cells carry a third of the
weight between them and the two stochastic cells carry another third.}
\label{tab:app-cells}
\end{table}

In the oracle diagnostic of \S\ref{sec:exp-ceiling}, the projected base-stock rule
with $\dep=1$ and realized future demand supplied as a point-mass forecast matches
the upper bound in all four deterministic cells. The fixed rule solves
\eqref{eq:marginal} against this degenerate predictive. Evaluating the shipped
adaptive-$\dep$ variant or a rule supplied with the true demand distribution
requires separate experiments.

The validation generator behind \S\ref{sec:exp-protocol} makes fresh-seed rebuilds
of the ten documented synthetic families, plus a retail-shaped family calibrated
only to the five training demands available to a policy. The composition
sweep, both proposer passes and the freeze draw their instances from that
generator and have no path to the benchmark's own. Candidates are inspected before
they run, and one is rejected if it refers to an instance's realized demand, to
the diagnostic switch that substitutes that demand for the forecast, to the
benchmark's own instance files, or to the per-period lead-time sequence; each
rejection is recorded. This static check targets accidental prohibited accesses;
its resistance to deliberate evasion has not been evaluated.

Table~\ref{tab:app-validation} reports the four configurations used to select
the final settings of the three components. Selection scores are measured on the
generator's instances, with the test set excluded from candidate scoring.
Validation instances carry no product text, so the shape layer is inactive and
its quality is outside this validation comparison. The harness checks this
inactivity by requiring the orders emitted with and without the layer to agree
entry by entry.

\begin{table}[t]
\centering
\scriptsize
\begin{tabular}{@{}lrrrrrl@{}}
\toprule
Forecaster/search & seed 0 & seed 1 & seed 2 & mean & sd & \\
\midrule
core / core & 0.61002 & 0.61009 & 0.60993 & 0.61001 & 0.00008 & frozen OR core \\
core / first pass & 0.61510 & 0.61483 & 0.61489 & 0.61494 & 0.00014 & first pass only \\
core / shipped & 0.61781 & 0.61739 & 0.61737 & 0.61752 & 0.00025 & \textbf{shipped} \\
pooled / shipped & 0.61839 & 0.61812 & 0.61813 & 0.61821 & 0.00015 & pooled, not shipped \\ 
\bottomrule
\end{tabular}
\caption{Validation scores of the four configurations that decided the freeze,
labelled by the forecaster and search settings. These scores support comparisons
within the validation generator; \S\ref{sec:exp} reports a separate test set. The
second row is the loop's first pass and the third is what ships, so the two
together show how the search component's gain divides between the loop's two
passes. The fourth scores highest and is not shipped, because its
forecaster pools across the instances of a cell and the official
one-instance-at-a-time runner cannot reproduce it; \S\ref{sec:exp-null} reports
its test-set contribution. The standard deviation across seeds measures Monte
Carlo noise from policy sampling on fixed data; uncertainty from sampling
validation instances is outside this measure.}
\label{tab:app-validation}
\end{table}

\subsection{Control arm and paired intervals}

The same-model online-adjustment control of Table~\ref{tab:control} was run in the
predecessor project.
Its results are preserved in the standing-context block supplied at the start of
every proposer round and retained with the experiment record. This control uses
one arm and one model, with roughly twenty instances per cell. At this sample
size, two of the six cells are more than two standard errors from zero: the real
deterministic cell at lead time four has a negative estimate, and the real
stochastic cell has a positive estimate.

The across-seed standard deviation of the shipped configuration is \abseedsd{}
and measures rollout Monte Carlo noise on fixed data. We assess sensitivity to
individual instances with the paired bootstrap. Every arm is scored on the same
instances, so the contrasts
are paired; we resample instances with replacement, keeping each instance's whole
row of arms together, and report the percentile interval of the paired mean
difference. These intervals are conditional on the fixed test set; their scope
is sensitivity to its instance composition, rather than generalization across
benchmarks. The measured contributions are distributed across instances:
the shape layer moves nearly every real-cell instance and no synthetic one, the
search layer nearly every stochastic-cell instance, and no contrast is dominated
by its ten largest movers.

\begin{table}[t]
\centering
\scriptsize
\setlength{\tabcolsep}{4pt}
\begin{tabular}{@{}lrrcrr@{}}
\toprule
Contrast against the frozen core & $n$ & $\Delta\Score$ & 95\% interval
& moved & top-10 \\
\midrule
Shape layer, all cells & 1320 & $+$0.0081 & $[+0.0058,\, +0.0105]$ & 0.44 & 0.20 \\
Search layer, all cells & 1320 & $+$0.0077 & $[+0.0061,\, +0.0093]$ & 0.33 & 0.17 \\
Both authored layers, all cells & 1320 & $+$0.0139 & $[+0.0115,\, +0.0164]$ & 0.62 & 0.14 \\
\quad shape, three real cells only & 600 & $+$0.0179 & $[+0.0128,\, +0.0230]$ & 0.97 & 0.20 \\
\quad search, two stochastic cells only & 440 & $+$0.0230 & $[+0.0185,\, +0.0277]$ & 0.99 & 0.17 \\
Pooled forecaster, on top of the shipped pair & 1320 & $+$0.0006 & $[-0.0000,\, +0.0013]$ & 0.98 & 0.08 \\%
 
\bottomrule
\end{tabular}
\caption{Paired bootstrap over the test instances. \emph{moved} is the share of
instances whose score changes at all, \emph{top-10} the share of total absolute
movement carried by the ten largest movers. The last row is measured
against the shipped pair, the others against the frozen core. Instance-level
scores come from a single seed, so a contrast here can differ slightly from the
three-seed mean in Figure~\ref{fig:ablation}. For the pooled forecaster, the
single-seed interval includes zero.}
\label{tab:app-paired}
\end{table}

\begin{figure}[t]
\centering
\includegraphics[width=0.56\textwidth]{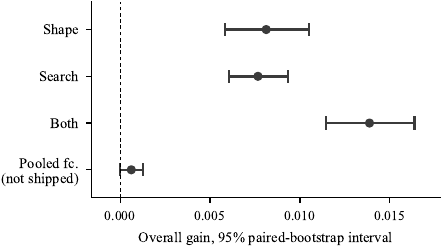}
\caption{Paired bootstrap intervals for the overall contributions, the same
quantities as Table~\ref{tab:app-paired}. The pooled forecaster's interval
includes zero; this configuration is retained as a diagnostic.}
\label{fig:app-paired-forest}
\end{figure}

The composition sweep identified the four first-pass gains as variants of the
same local-rate normalization in the stochastic search. The second pass added
the cap; the final search component combines these two changes.

The alternative forecaster adds \abrshape{} with the seasonal layer and
\abrsearch{} with the search layer, indicating that its gain depends on the
accompanying components. The negative seasonal-by-search interaction in
\S\ref{sec:exp-abl} indicates overlapping gains from the two layers.
Configurations with a pooled forecaster are excluded from submission to preserve
the per-instance protocol.

On the two cells with a deterministic lead time of four, the hindsight-tuned
constant level of \S\ref{sec:exp-ceiling} exceeds every one of the
\npublished{} public baselines. That section reports the strongest baseline's
per-cell scores.

\subsection{Design-time budget and reproducibility}

The fifty-two proposals of \S\ref{sec:exp-null} are all retained, and the second
pass is recorded round by round: twelve of twenty-eight improved on the screening
subset, four earned a full-validation re-score, one failed to run, and the best of
the four was frozen. Those two passes and the shape elicitation of
\S\ref{sec:exp-repro} are the whole design-time budget, elicitation being most of
it; the products it covers appear in all three real cells, so three samples for
each of the \shapeitems{} products covers the benchmark's real half once. Rows of
Figure~\ref{fig:ablation} come from runs that take a cell's instances in one
batch, the path the design-time search used; moving to the one-instance-at-a-time
path a submission is scored on changes the shipped configuration by \execgap{}
overall and reorders nothing.

\begin{figure}[t]
\centering
\includegraphics[width=0.55\textwidth]{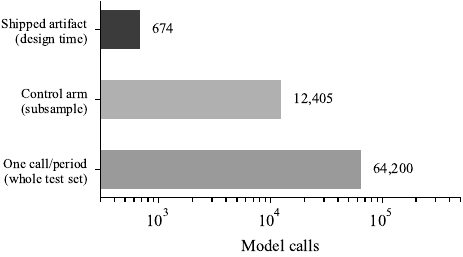}
\caption{Model calls, log scale. The top bar counts all calls used to create the
shipped artifact: shape elicitation and both proposer passes. The middle bar
counts retained responses for the same-model online-adjustment control of
\S\ref{sec:exp-abl}, reported as a per-period model call on a \ctrln{}-instance subsample. This
experimental cost is shown separately from artifact construction. The bottom
bar gives the calls required to run such a control once per period over the
whole test set. The shipped policy makes no model calls during operation.}
\label{fig:app-budget}
\end{figure}

The control arm's budget is based on \armcalls{} retained responses, roughly
twice \ctrln{} instances times a horizon. This count indicates either coverage
beyond the reported \ctrln{} instances or multiple calls in some periods.
The predecessor project's logs leave these possibilities unresolved, so the
budget uses the retained count. Per-cell deltas use the reported runs, while
the control arm's cost remains uncertain by a factor of about two.

The reproduction of \S\ref{sec:exp-repro} produces the orders, scores them with
the benchmark's own evaluator and replays them through the official policy
protocol, one instance at a time exactly as that evaluator drives a policy. The
replay reproduces the submitted orders with maximum deviation zero, so the shipped
configuration carries no cross-instance dependence. Design-time replies are
retained and a verification mode checks that every one a run needs is already
present. This check verifies that a repeat run can use cached replies throughout,
with no network calls.

The proposer context includes four test-set diagnostics: the score of a
hindsight-tuned global order-up-to level; that the deterministic rule reaches the
upper bound when fed the realized path; that the stochastic cells are furthest
from it; and the
per-cell deltas and standard errors of the control arm. The record carries no test
score of any candidate, every one having been admitted or rejected on validation
alone.

\subsection{Numerical checks on the appendix results}

We numerically check three results proved in Appendix~\ref{app:proofs}.
For Lemma~\ref{app:lem:jacobian} we obtain
the raw-space density by differencing the transformed CDF and match it to
$q(\tf(\dem))\tf'(\dem)$, confirm that the uncorrected version integrates to
$27.66$ instead of one, exhibit an observation on which dropping the correction
reverses which member wins, and check the closed form for $\E\Delta_\per$ to four
digits over two hundred periods. For Proposition~\ref{app:prop:supply} we run
random trajectories under both accountings, including the degenerate ones in which
every lead time is zero and in which every order is destroyed, and find no period
where the identity fails or where the retaining accounting leaves the indicator
identified. For Proposition~\ref{prop:shape-bound} we check part (i) against
random loss sequences, check the floor and the per-period bound of part (ii)
against sequences built to satisfy the $\delta$ hypothesis, and then remove the
hypothesis and drive the flat variant's weight below the smallest positive
representable number, demonstrating the role of that hypothesis.

For the cost-structure results we solve the zero-lead-time dynamic program exactly
by backward induction over a bounded integer state. The computed solutions
confirm that $V_\per$ is flat
on $[0,\base^*_\per]$ and non-increasing and that $\max\{\onh,\base^*_\per\}$ is
optimal at every state under Assumption~\ref{app:as:levels}. The checks also cover
the three falling-level cases discussed after Theorem~\ref{app:thm:zero-option} and
the integer difference of Corollary~\ref{app:cor:integer} and its argmax over four
hundred random demand laws. They also verify that $\Phi$ is non-increasing, that it
crosses zero once and that the crossing is the least maximizer of its integral, at
$\dep = 0$ and at $\dep > 0$.

The checks cover three boundary cases. On $\dem_1$ uniform on $\{2,5\}$ and
$\dem_2$ on $\{1,4\}$ at $\crat = 1/2$ the smallest integer fractile selection is
$(2,1)$ and falls while the monotone real selection $(2,2)$ exists, which is why
Corollary~\ref{app:cor:integer} cannot be had by invoking the theorem with that
selection; the conclusion nonetheless holds at every state there. On a shortfall
supported below zero, the ordinary case of a pipeline that already covers the
demand it is responsible for, $\Phi(0) < 0$, the maximizer of $\int_0^{\ord}\Phi$
over $\ord \ge 0$ is exactly $\{0\}$ and the crossing set of
Proposition~\ref{app:prop:marginal}(iv) is empty, selecting the zero-order
branch; the crossing on the whole line sits at $-3$ against $\ord^0 = 0$,
which is the difference part (vi) records. Finally the pathwise identity
$\Short^{\mathrm{naive}} = \Short^{\mathrm{proj}} + \Lost$ is checked over four
thousand random pipelines, with equality exactly on the paths where nothing was
lost, and is then run past the end of the horizon, where in-transit units that
never land leave a residual of $-7$ and the identity fails. That is the case
$\per+\lead\leq\Tend$ excludes. In the paper's one-based indexing, the policy
likewise places no order once $\per+\lead>\Tend$.

\end{document}